\documentclass[11pt]{article}

\usepackage[a4paper,margin=0.85in]{geometry}
\usepackage[T1]{fontenc}
\usepackage{stix}
\usepackage[numbers,sort&compress]{natbib}
\usepackage{mathtools}
\usepackage{amsthm}
\usepackage{bm}
\usepackage{graphicx}
\usepackage{booktabs}
\usepackage{makecell}
\usepackage{array}
\usepackage{multirow}
\usepackage{subcaption}
\usepackage{placeins}
\usepackage[dvipsnames]{xcolor}
\usepackage[colorlinks=true,linkcolor=MidnightBlue,citecolor=MidnightBlue,urlcolor=MidnightBlue]{hyperref}
\usepackage[nameinlink,capitalize,noabbrev]{cleveref}
\graphicspath{{figs/}}

\newtheorem{theorem}{Theorem}
\newtheorem{lemma}[theorem]{Lemma}
\newtheorem{proposition}[theorem]{Proposition}
\newtheorem{corollary}[theorem]{Corollary}
\theoremstyle{definition}
\newtheorem{assumption}[theorem]{Assumption}
\newtheorem{definition}[theorem]{Definition}
\theoremstyle{remark}
\newtheorem{remark}[theorem]{Remark}
\crefname{assumption}{assumption}{assumptions}
\Crefname{assumption}{Assumption}{Assumptions}

\newcommand{\Uh}{\mathcal{U}}
\newcommand{\Uzero}{\mathcal{U}_0}

\newcommand{\rcut}{r_{\rm cut}}
\newcommand{\Scal}{\mathcal{S}}

\newcommand{\Hhom}{H^{\rm hom}}
\newcommand{\Hess}{H}
\newcommand{\Fhom}{F}
\newcommand{\Id}{I}

\newcommand{\Tr}{\operatorname{Tr}}
\newcommand{\dist}{\operatorname{dist}}

\newcommand{\BZ}{\mathcal{B}}
\newcommand{\dd}{\,{\rm d}}

\title{Local Surrogates for Harmonic Vibrational Entropy in Multilattices}

\author{%
Tina Torabi\\
Department of Mathematics\\
University of British Columbia\\
\texttt{ttorabi@math.ubc.ca}
\and
Jiale Linghu\\
School of Mathematics and Statistics\\
Xidian University
\and
Yangshuai Wang\\
Department of Mathematics\\
National University of Singapore\\
\texttt{yswang@nus.edu.sg}\\
\emph{Corresponding author}
}
\date{}

\begin{document}
\maketitle

\begin{abstract}
Harmonic vibrational entropy is a key finite-temperature contribution to defect thermodynamics, but direct evaluation by dense Hessian diagonalization scales cubically with atom count and is too costly for supercell convergence, migration-path sampling, and high-throughput defect studies. We develop local surrogate models for harmonic entropy in multilattices, including semiconductors, ordered alloys, and multispecies crystals with multi-atom bases and internal degrees of freedom. Unlike Bravais lattices, multilattices contain internal-shift degrees of freedom and optical phonon modes coupled to acoustic strain; entropy models must therefore resolve sublattice and species labels. For finite-range or screened atomistic models, we prove sublattice-resolved locality and cutoff-error estimates that justify replacing the global entropy calculation by a local, symmetry-respecting regression problem with controlled truncation error. This turns vibrational entropy from a global spectral calculation into a reusable local site model with linear evaluation cost at fixed cutoff. Numerical tests confirm the predicted locality behavior and show that sublattice/species-resolved surrogates achieve accurate regression, transfer across supercell sizes, and linear-scaling evaluation on Stillinger--Weber Si and CdTe benchmarks. The resulting method enables repeated harmonic-entropy evaluations in multispecies defect calculations while retaining explicit stability, truncation, and surrogate-error controls.
\end{abstract}

\noindent\textbf{Keywords:} Vibrational entropy, Multilattice, Crystalline defects, Atomic cluster expansion, Surrogate model, Thermodynamic limit

\bigskip
\section{Introduction}
\label{sec:introduction}

Vibrational entropy is a leading-order finite-temperature contribution to defect thermodynamics in crystals~\cite{Maradudin1971Lattice,Wallace1972Thermodynamics,Catlow1991Defects}. In the harmonic approximation it is obtained from the positive spectrum of a supercell Hessian, equivalently from a log-determinant after removing rigid translations. For one configuration this is routine; in practical defect studies it becomes a repeated global spectral calculation. Entropy differences enter formation free energies, equilibrium defect concentrations, transition-state prefactors for diffusion, migration barriers sampled by nudged-elastic-band images, and thermodynamic-integration or uncertainty-sampling calculations~\cite{Vineyard1957FrequencyFactors,Henkelman2000NEB,KirkwoodTI1935,Frenkel1984FreeEnergy,Grabowski2009AnharmonicFreeEnergy,Foiles1994VacancyEntropy,Glensk2014Vacancies}. The dense eigendecomposition scales as $O(N^3)$ in the number of atoms, so harmonic entropy can become the limiting cost in supercell convergence, alloy sampling, and high-throughput defect calculations.

This bottleneck is especially important for multilattices. Many technologically important crystals are multilattices rather than one-atom Bravais crystals: semiconductors, ordered intermetallics, B2 and zinc-blende compounds, wurtzite materials, and multispecies crystals with a basis. Their mechanical and thermodynamic response couples long-wavelength elastic strain to relative sublattice relaxations, often called internal shifts~\cite{KohlhoffGumbsch1991Multilattice,ELuYang2006Multilattice,OlsonOrtnerWangZhang2023MultilatticeDislocations}. In lattice dynamics these internal degrees of freedom produce optical phonon branches and acoustic--optical coupling~\cite{Maradudin1971Lattice,Wallace1972Thermodynamics}. A local entropy model for such materials therefore cannot treat all atoms as samples of a single Bravais environment. It must resolve chemical species, sublattice identity, and internal-shift contributions to the vibrational free energy.

Existing computational routes address different pieces of this problem but do not provide a reusable local entropy model with quantified cutoff and fitting errors. Direct phonon calculations and ab-initio defect thermodynamics extract entropies from supercell spectra~\cite{Foiles1994VacancyEntropy,Glensk2014Vacancies,Marinica2013ProvingFormationFiniteSize,Gillan1989Defects}; finite-temperature free-energy calculations for complex defect mechanisms face related sampling and anharmonicity costs~\cite{SwinburneMarinica2018FreeEnergyBarriers}; temperature-dependent effective potentials and self-consistent phonon methods renormalize force constants but retain a global spectral step~\cite{Esfarjani2008LatticeDynamics,Hellman2013SCAILD,Errea2014SCHA}; stochastic log-determinant methods accelerate individual trace or determinant estimates~\cite{Hutchinson1990Trace,Bai1996BoundsLogDet,Ubaru2017FastEstimation,HanMalioutov2017LogDetApprox,Higham2008Functions} but do not produce transferable site models. Machine-learning surrogates have been used successfully in point-defect and migration studies~\cite{Lapointe2020VibrationalEntropyML,Lapointe2022StrainDependentVibrational,Goryaeva2020GapForVacancies}; more broadly, active-learning, machine-learned MM, and hybrid ab-initio--machine-learning workflows now target large-scale defect and alloy simulations, adaptive QM/MM coupling, and fine-tuned universal interatomic potentials~\cite{HodappShapeev2020InOperando,HodappShapeev2021RandomAlloys,ChenOrtnerWang2022QMMMML,Grigorev2023WDislocationBinding,WangKermodeOrtnerZhang2024QMMM,LiuZengLuoWangZhaoXu2026FineTuningUMLIP}. Yet their locality is usually imposed by modelling choice rather than derived as a truncation-controlled target.

The contribution of this paper is to turn harmonic entropy in stable finite-range multilattices into such a local target. Building on the Bravais-lattice entropy locality theory of~\cite{TorabiWangOrtner2024Entropy} and the atomistic analysis of multilattice Green's functions~\cite{OlsonOrtnerWangZhang2023MultilatticeDislocations}, and in the broader context of lattice-Green-function and far-field boundary treatments of extended defects and general crystals~\cite{Trinkle2008LGF,YasiTrinkle2012LGF,BraunOrtnerWangZhang2025FarFieldBC}, we prove a sublattice-resolved locality theorem and a corresponding cutoff-error estimate under full block phonon stability and an admissible local stiffness perturbation. In computational terms, the result justifies replacing the global Hessian log-determinant by a symmetry-respecting local regression problem whose truncation error is controlled by the cutoff radius. This gives a route to repeated entropy evaluations with the same locality logic that makes finite-range defect mechanics computationally tractable. The analysis is stated for finite-range or screened effective models; long-range polar effects require a separate splitting or screening treatment.

This locality result is paired with a practical surrogate construction. We use the atomic cluster expansion (ACE) as a systematic basis for local atomistic functions~\cite{Drautz2019ACE,Bachmayr2022ACECompleteness,Bachmayr2020Symmetric,Hashemi2017FullProof}, alongside simpler local-polynomial tests. This places the construction within the broader local-environment representation framework used by SOAP, GAP, moment-tensor, and related machine-learning potentials~\cite{BartokKondorCsanyi2013SOAP,BartokPayneKondorCsanyi2010GAP,ShapeevMTP2016}. The descriptors are species- and sublattice-resolved, because the theory identifies the entropy contribution with inequivalent basis sites rather than with an undifferentiated local environment. The resulting error decomposition separates truncation, basis approximation, and statistical estimation, allowing the cutoff radius and regression complexity to be chosen together.

Numerical experiments test both the analytical rates and the resulting surrogate calculation. Synthetic multilattice models in one, two, and three dimensions show the predicted locality and truncation behaviour over the tested ranges. On diamond Si under the Stillinger--Weber potential, the symmetry-adapted ACE site model remains accurate under both randomized site splits and configuration-level splits, transfers to an independent larger supercell, and reduces repeated evaluation to a linear-scaling local calculation at fixed cutoff. On zinc-blende CdTe, a species-resolved local-polynomial ablation shows that retaining species/sublattice labels is essential for accuracy. Together these results demonstrate a route from global vibrational spectra to reusable local entropy surrogates for multispecies defect calculations.

This paper is structured as follows. \Cref{sec:setting} introduces the multilattice model class, with emphasis on the acoustic, optical, and internal-shift variables that distinguish multilattice defect mechanics from the Bravais setting. \Cref{sec:site-entropy} proves the site-entropy locality and truncation estimates that make a local entropy surrogate meaningful. \Cref{sec:ace} formulates the species- and sublattice-resolved surrogate and its error decomposition. \Cref{sec:numerics} reports the synthetic rate tests and the Si/CdTe benchmark calculations. \Cref{sec:conclusions} summarizes implications for repeated finite-temperature defect calculations and outlines the remaining limitations.

\section{Multilattice Setting}
\label{sec:setting}

This section specifies the atomistic model used in the analysis and in the surrogate construction. The goal is to keep the notation close to what is used in supercell calculations: a crystal is represented by primitive cells, each cell contains a fixed set of basis sites, and the Hessian is assembled from a finite local stencil. The only additional structure, compared with a Bravais lattice, is that the basis atoms can move relative to one another. These relative motions are the internal shifts that generate optical phonon modes and couple to acoustic strain. They are the reason that harmonic entropy in multilattices must be decomposed by species and sublattice.

We model a multilattice as a Bravais lattice decorated by a finite set of basis sites in each primitive cell. This setting covers the structures used throughout semiconductor, alloy, and intermetallic defect modelling: diamond, zinc-blende, rocksalt-type, B2, wurtzite, and more complex multi-sublattice motifs. \Cref{fig:multilattice-structures} displays four representative two-sublattice cases. The Bravais-lattice analysis of~\cite{TorabiWangOrtner2024Entropy} does not register the internal-shift variables; the multilattice setting below is the corresponding model class for finite-range or screened multilattice models.

\begin{figure}[t]
\centering
\includegraphics[width=\linewidth]{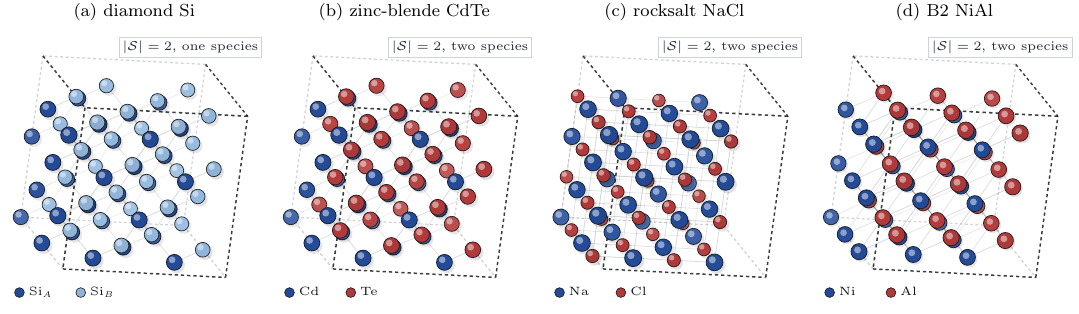}
\caption{Representative two-sublattice crystal structures ($|\Scal|=2$): (a)~diamond Si, with one chemical species on two FCC sublattices; (b)~zinc-blende CdTe, with chemically distinct cation and anion sublattices; (c)~rocksalt NaCl, shown as a structural motif whose unscreened Coulomb physics is outside the finite-range theorem; (d)~B2 (CsCl-type) NiAl, with two species on interpenetrating simple-cubic sublattices. Two shades of one colour mark the two Si sublattices in panel (a), and two colours mark chemically distinct species in panels (b)--(d). The dashed lines mark the conventional cubic unit cell.}
\label{fig:multilattice-structures}
\end{figure}

We now specify the geometric and operator notation. Throughout the paper $d\in\{1,2,3\}$ is the spatial dimension.

\subsection{Geometry and Displacement Variables}
\label{sec:geometry-displacements}

We first fix the reference crystal and the displacement variables. Let $A\in\mathbb{R}^{d\times d}$ be a non-singular matrix whose columns generate the Bravais lattice
\[
  \Lambda=A\mathbb{Z}^d\subset\mathbb{R}^d .
\]
The finite set $\Scal$ labels the inequivalent basis sites inside one primitive cell. When the basis atoms have the same chemical species, as in diamond Si, the label $\alpha\in\Scal$ is a sublattice label; when they are chemically distinct, as in CdTe or NiAl, it is also a species label. For each $\alpha\in\Scal$, the vector $p_\alpha^{\rm ref}\in\mathbb{R}^d$ is the reference position of that basis site inside the primitive cell. One of these reference shifts may be chosen as zero by translating the origin. The multilattice is the disjoint union
\begin{equation}
  \mathcal{M} = \bigcup_{\alpha \in \mathcal{S}} \left( \Lambda + p_\alpha^{\rm ref} \right),
  \label{eq:multilattice-def}
\end{equation}
Every atom is therefore indexed by a pair $(\ell,\alpha)\in\Lambda\times\Scal$. Here $\ell$ identifies the primitive cell and $\alpha$ identifies the basis site in that cell. The number $|\Scal|$ is fixed throughout the analysis: $|\Scal|=1$ recovers the Bravais case, $|\Scal|=2$ covers the examples in \cref{fig:multilattice-structures}, and larger values cover multi-sublattice structures such as wurtzite, perovskites, fluorites, and Heusler alloys.

A deformed configuration is described by a displacement field
\[
  u:\Lambda\times\Scal\to\mathbb{R}^d,
  \qquad
  (\ell,\alpha)\mapsto u_\alpha(\ell),
\]
where $u_\alpha(\ell)$ is the displacement of the atom with cell index $\ell$ and basis label $\alpha$. The deformed atomic position is
\begin{equation}
  y_\alpha(\ell) = \ell + p_\alpha^{\rm ref} + u_\alpha(\ell),
  \qquad \ell\in\Lambda,\ \alpha\in\mathcal{S}.
  \label{eq:deformation}
\end{equation}
In computations, the same field contains both the long-wavelength elastic displacement and the relative shifts between sublattices. We separate these two contributions in \cref{sec:block-structure}; for the moment, $u$ denotes the full multilattice displacement field $\{u_\alpha(\ell)\}_{(\ell,\alpha)\in\Lambda\times\Scal}$.

The basic displacement space is
\begin{equation}
  \Uh:=\{u:\Lambda\times\Scal\to\mathbb{R}^d:\,\|u\|_{L^\infty}<\infty\},
  \qquad
  \|u\|_{L^\infty}:=\sup_{\ell,\alpha}|u_\alpha(\ell)|,
\end{equation}
the space of bounded multilattice displacements. We also use the compactly supported subspace
\begin{equation}
  \Uzero:=\{u\in\Uh:\,\mathrm{supp}\,u\text{ is finite}\}.
\end{equation}
Compact-core perturbations model deliberately truncated defect cores or strictly localized changes of the stiffness tensor. Relaxed point defects such as vacancies, interstitials, and anti-sites generally induce elastic far-fields, while dislocations or inclusions generate algebraically decaying fields that lie in $\Uh$ but not in $\Uzero$. The sharp locality theorem below is stated for compact-core perturbations and for coefficient fields satisfying the admissibility condition in \cref{def:admissible-perturbation}; algebraic far-fields are recorded separately through the explicit convolution bound of \cref{lem:localized-perturbation}.

\subsection{Finite-Range Energy and Difference Stencils}
\label{sec:finite-range-structure}

The harmonic entropy depends on the Hessian of the atomistic energy. We assume that this Hessian is assembled from a finite stencil of local differences, as in standard finite-range empirical or screened effective models. Two types of differences are needed in a multilattice. The first is the \emph{intra-sublattice difference} along a lattice vector $j\in\Lambda$,
\begin{equation}
  (D_j u)_\alpha(\ell):=u_\alpha(\ell+j)-u_\alpha(\ell),
\end{equation}
which compares two atoms on the same sublattice and detects the usual elastic strain. The second is the \emph{sublattice-shift difference} between sublattices $\alpha$ and $\beta$ at lattice offset $j$,
\begin{equation}
  (\widetilde D_{j;\alpha\beta} u)(\ell):=u_\beta(\ell+j)-u_\alpha(\ell),
\end{equation}
which compares atoms on different sublattices and detects the internal shift. We write $\mathcal{D}_q$ for one finite-range bond leg of the Hessian stencil. Each $\mathcal{D}_q$ is a finite linear combination of operators of the form $D_j$ and $\widetilde D_{j;\alpha\beta}$; the index $q$ simply enumerates this finite list of bond legs. The bond-difference seminorm used below is
\begin{equation}
  \|Du\|_{\ell^\infty}
  :=\sup_{\mathcal{D}_q}\sup_{\ell\in\Lambda}|\mathcal{D}_q u(\ell)|,
  \label{eq:bond-inf-norm}
\end{equation}
where the first supremum is over the finite stencil. This norm measures local strain and internal shift rather than absolute displacement. Rigid translations have zero bond seminorm and therefore do not affect the Hessian or the entropy.

The atomistic energy is written as a sum of finite-range site contributions,
\begin{equation}
  \mathcal{E}^{\rm a}(u)
  :=\sum_{(\ell,\alpha)\in\Lambda\times\Scal}V_{\ell,\alpha}\!\left(\big\{y_\beta(m)-y_\alpha(\ell):(m,\beta)\in\Lambda\times\Scal,\,|m-\ell|\leq r_V\big\}\right).
  \label{eq:atomistic-energy}
\end{equation}
Here $r_V>0$ is the interaction radius, and $V_{\ell,\alpha}$ is the site potential associated with the atom $(\ell,\alpha)$. The argument of $V_{\ell,\alpha}$ is the finite list of relative position vectors from this atom to its neighbours within distance $r_V$. In a homogeneous crystal the site potential is periodic in the cell index and the resulting Hessian is translation invariant on $\Lambda$.

The second variation of $\mathcal{E}^{\rm a}$ at displacement $u$ defines the harmonic Hessian
\[
  \langle \Hess(u)v,w\rangle:=\delta^2\mathcal{E}^{\rm a}(u)[v,w],
  \qquad v,w\in\Uzero .
\]
We identify $\Hess(u)$ with the corresponding bounded symmetric operator on $\ell^2(\Lambda\times\Scal;\mathbb{R}^d)$. All Hessians entering the entropy formula are understood in mass-weighted coordinates; using an equivalent unweighted acoustic/internal-shift basis changes only fixed block constants and does not change the decay estimates. The homogeneous Hessian is
\[
  \Hhom:=\Hess(0),
\]
the Hessian at the reference multilattice. By finite range, the kernel of $\Hess(u)$ is supported on bond pairs within the interaction range and is generated by the finite set of difference operators introduced above.

\begin{assumption}[Finite range and smoothness]
\label{ass:smoothness}
The atomistic energy is generated by a finite-range site potential with bounded derivatives up to fourth order in a neighborhood of the homogeneous multilattice state. The interaction range and the derivative bounds are independent of the supercell size.
\end{assumption}

\Cref{ass:smoothness} is the standard regularity hypothesis used throughout the atomistic-to-continuum literature~\cite{HudsonOrtner2014A2C,Hudson2014Stable,LuskinOrtner2013Acta}. The empirical Stillinger--Weber potentials~\cite{StillingerWeber1985Silicon} used in \cref{sec:numerics} serve as finite-range benchmark models on the sampled displacement geometries, which stay away from cutoff and instability events; the same local interpretation applies to EAM and Tersoff-type benchmarks~\cite{DawBaskes1984EAM,Tersoff1989Multicomponent}. We write
\begin{equation}
  \|\partial^{(\leq k)}V\|:=\max_{0\leq j\leq k}\sup_{y}\big|\partial^{j}V(y)\big|
  \label{eq:dV-norm}
\end{equation}
for the maximum supremum norm of the partial derivatives of the site potential up to order $k$, where the supremum is taken over an open neighbourhood of the homogeneous configuration.

Let
\begin{equation}
  \mathcal{R}:=\{j\in\Lambda:\Hhom(j)\neq 0\},
  \qquad
  \rcut^{(\rm hom)}:=\max\{|j|:j\in\mathcal{R}\},
  \label{eq:rcut-hom-def}
\end{equation}
denote the finite stencil of the homogeneous Hessian and its diameter. Both are finite by \cref{ass:smoothness} and independent of the supercell size. The constant $\rcut^{(\rm hom)}$ enters only through fixed stencil-dependent prefactors. We also use the discrete $H^1$ seminorm
\begin{equation}
  \|v\|_{a1}^2
  :=\sum_{j\in\mathcal{R}}\sum_{\ell\in\Lambda}\sum_{\alpha\in\Scal}|(D_j v)_\alpha(\ell)|^2,
  \qquad v\in\Uzero,
  \label{eq:a1-seminorm}
\end{equation}
which is finite for every compactly supported test displacement $v\in\Uzero$.

\subsection{Acoustic--Optical Block Structure}
\label{sec:block-structure}

The key multilattice feature is the separation between cell-average motion and internal-shift motion. For a displacement $u$, define its cell average
\[
  \bar u(\ell):=|\Scal|^{-1}\sum_{\alpha\in\Scal}u_\alpha(\ell).
\]
The decomposition $u=\bar u+(u-\bar u)$ identifies $\Uh$ with the direct sum $\Uh^{(0)}\oplus\Uh^{(p)}$, where the superscript $0$ denotes the acoustic or cell-average component and $p$ denotes the internal-shift or optical component. The unweighted average is used as a convenient fixed projector; the analogous mass-weighted decomposition has the same algebraic form with rescaled constants. A symmetric operator $K$ on $\Uh$ then takes the $2\times2$ block form
\begin{equation}
  K=\begin{pmatrix} K_{00} & K_{0p}\\ K_{p0} & K_{pp}\end{pmatrix},
  \qquad K_{0p}=K_{p0}^*,
  \label{eq:block-decomp}
\end{equation}
where $K_{00}$ couples acoustic components, $K_{pp}$ couples internal-shift components, and $K_{0p}=K_{p0}^*$ is the acoustic--optical coupling. This is the block notation of~\cite{OlsonOrtnerWangZhang2023MultilatticeDislocations}. The site index $(\ell,\alpha)$ and the block index $(0,p)$ are related by a fixed cell-internal change of basis. We use the site representation for the ACE descriptors in \cref{sec:ace}, and the block representation for the Green-function estimates.

Let $P_0$ be the cell-average projector and let $P_p:=\Id-P_0$ be the internal-shift projector. In Fourier variables, an intra-sublattice difference has symbol $(e^{-i\xi\cdot j}-1)(P_0+P_p)$ and therefore contributes an $O(|\xi|)$ factor near $\xi=0$. A sublattice-shift difference has symbol $E_{\alpha\beta}e^{-i\xi\cdot j}-E_{\alpha\alpha}$ in the original sublattice basis, where $E_{\alpha\beta}$ is the matrix unit mapping the $\beta$ sublattice component to the $\alpha$ component. On the acoustic subspace, where all sublattices move together, this symbol again contains the factor $e^{-i\xi\cdot j}-1=O(|\xi|)$; on the internal-shift subspace it remains uniformly bounded. This elementary calculation explains why finite differences regularise the acoustic singularity but must still track the acoustic, mixed, and optical blocks separately.

For a translation-invariant operator $K$ with kernel $K(\ell,m)=K(\ell-m)$, define its lattice Fourier symbol by
\begin{equation}
  \widehat K(\xi)=\sum_{j\in\Lambda}e^{-i\xi\cdot j}K(j),\qquad \xi\in\BZ,
  \label{eq:fourier-symbol}
\end{equation}
where $\BZ$ is the Brillouin zone of $\Lambda$. Applied to the homogeneous Hessian $\Hhom$, the symbol $\widehat\Hhom(\xi)$ is a Hermitian, positive-semidefinite $|\Scal|d\times|\Scal|d$ matrix and has the block form
\begin{equation}
  \widehat\Hhom(\xi)=\begin{pmatrix}\widehat{\mathbb{C}}_{\rm a}(\xi) & \widehat{\mathbb{C}}_{\rm ap}(\xi)\\ \widehat{\mathbb{C}}_{\rm ap}(\xi)^* & \widehat{\mathbb{C}}_{\rm p}(\xi)\end{pmatrix}.
  \label{eq:hom-symbol-block}
\end{equation}
Here $\widehat{\mathbb{C}}_{\rm a}$ is the acoustic block, $\widehat{\mathbb{C}}_{\rm p}$ is the optical block, and $\widehat{\mathbb{C}}_{\rm ap}$ is the acoustic--optical coupling. Near $\xi=0$, the acoustic block behaves like $|\xi|^2$ times an elastic stiffness tensor, the optical block remains uniformly positive for a stable multilattice, and the coupling vanishes linearly in $|\xi|$. This three-rate structure is the source of the multilattice Green-function hierarchy of~\cite{OlsonOrtnerWangZhang2023MultilatticeDislocations} and is the structural reason the Bravais entropy-locality proof does not transfer without modification.

Throughout the paper an operator $K$ on $\ell^2(\Lambda\times\Scal;\mathbb{R}^d)$ is identified with its kernel $\big(K_{(\ell,\alpha),(m,\gamma)}\big)_{\ell,m\in\Lambda,\,\alpha,\gamma\in\Scal}$, where each entry $K_{(\ell,\alpha),(m,\gamma)}\in\mathbb{R}^{d\times d}$ is the Cartesian-component block coupling the displacement at site $(\ell,\alpha)$ to the displacement at site $(m,\gamma)$. The notation $|K_{(\ell,\alpha),(m,\gamma)}|$ denotes the operator (spectral) norm of this block; equivalent estimates hold for any other matrix norm by norm equivalence in fixed dimension $d$. All kernel-decay statements in \cref{lem:hom-kernel,lem:localized-perturbation,lem:block-resolvent} and in \cref{thm:locality} are understood in this convention, with hidden constants depending on $d$ but not on $\ell,m,\alpha,\gamma$.

\subsection{Stability and Admissible Perturbations}
\label{sec:stability-admissibility}

The locality theorem is a statement about stable configurations near a stable reference crystal. In computational terms, this means that the homogeneous phonon spectrum must have the expected acoustic zero modes and a positive optical gap, and that the defect or distortion must not move the Hessian out of this stable phonon basin. We express this requirement in the acoustic--optical block variables introduced above. The condition is stronger than positivity of the acoustic and optical diagonal blocks separately, because the acoustic--optical coupling can otherwise close the acoustic Schur complement.

\begin{assumption}[Homogeneous multilattice stability]
\label{ass:stability}
There exist constants $\gamma_a>0$ and $\gamma_{\rm p}>0$ such that the full homogeneous symbol satisfies the block-coercivity bound
\begin{equation}
  \widehat\Hhom(\xi)
  \succeq
  \gamma_a\big(|\xi|^2P_0+P_p\big),
  \qquad \xi\in\BZ,
  \label{eq:acoustic-stab}
\end{equation}
and its optical-block Fourier symbol $\widehat{\mathbb{C}}_{\rm p}$ in~\eqref{eq:hom-symbol-block} satisfies
\begin{equation}
  \widehat{\mathbb{C}}_{\rm p}(\xi)\succeq\gamma_{\rm p}\Id,\qquad\xi\in\BZ,
  \label{eq:optical-stab}
\end{equation}
where $\succeq$ denotes the partial order on Hermitian matrices.
\end{assumption}

\Cref{ass:stability} is the multilattice analogue of the phonon stability condition that already appears in the Bravais theory and in multilattice elasticity~\cite{KohlhoffGumbsch1991Multilattice,OlsonOrtnerWangZhang2023MultilatticeDislocations}. The form \eqref{eq:acoustic-stab} is deliberately a full-block condition: it is equivalent to positivity of the Schur complement of the optical block and prevents the $O(|\xi|)$ acoustic--optical coupling from destroying the acoustic lower bound. The optical constant $\gamma_{\rm p}$ in~\eqref{eq:optical-stab} controls the energy cost of internal shifts. The constants in \cref{thm:locality,thm:truncation} depend on the full stability margin, the optical gap, acoustic--optical coupling bounds, the conditioning of the block projectors, the finite stencil, and $|\Scal|$; for any fixed positive margins these parameters affect the prefactor but not the algebraic decay exponent.

\begin{remark}
The locality bound holds uniformly for fixed positive stability margins, but the prefactor deteriorates near acoustic loss of stability, optical softening, or ill-conditioned acoustic--optical splitting. Thus the truncation radius required for a fixed target accuracy may grow as $\gamma_{\rm p}\to0^+$ or as the Schur-complement margin closes. The analysis is most informative away from soft-mode boundaries; the examples in \cref{fig:multilattice-structures} should be interpreted as representative stable multilattices rather than as a statement that every empirical parameterisation is uniformly far from such boundaries.
\end{remark}

We next identify the class of distorted configurations for which the same stable linearization remains valid. The neighbourhood is written in terms of the bond-difference seminorm \eqref{eq:bond-inf-norm}, since the Hessian responds to local strains and internal shifts rather than to rigid translations.

\begin{definition}[Stability neighborhood]
\label{def:stab-nbhd}
For $\delta>0$, let $\mathcal{B}_\delta:=\{u\in\Uh:\|Du\|_{\ell^\infty}\leq\delta\}$. A sufficient stability neighborhood associated with \cref{ass:stability} is the set $\mathcal{B}_{\delta_*}$ where
\begin{equation}
  \delta_* :=
  c_*\,\frac{\gamma_{\rm st}}{\|\partial^{(\leq 4)}V\|}
  \label{eq:delta-star}
\end{equation}
and $c_*>0$ is a fixed dimension- and stencil-dependent constant. Here $\gamma_{\rm st}$ denotes the combined stability margin in \cref{ass:stability}, including the full-block coercivity and optical-gap margins.
\end{definition}

For computation, \cref{def:stab-nbhd} says that all configurations used to define site entropies or train the surrogate stay in the same stable harmonic regime as the homogeneous reference. Mathematically, $\mathcal{B}_{\delta_*}$ is a conservative bond-strain neighbourhood in which the normalized Hessian perturbation remains small. In particular, $u\in\mathcal{B}_{\delta_*}$ implies
\[
  \|(\Hhom)^{-1/2}(\Hess(u)-\Hhom)(\Hhom)^{-1/2}\|_{\ell^2\to\ell^2}\leq 1/2,
\]
so the matrix logarithm used in \cref{sec:site-entropy} is uniformly controlled. The constant $c_*$ is chosen after the homogeneous stencil, stability margins, and admissibility constants have been fixed. If $A(u):=\max_q\|A_q(u)\|_{\ell^2\to\ell^2}$ denotes the largest bond-stiffness perturbation size, the same choice also ensures the kernel-majorant smallness
\begin{equation}
  C_{\rm hom}A(u)\leq \theta r_0,\qquad 0<\theta<1,
  \label{eq:kernel-smallness}
\end{equation}
for the contour radius $r_0=3/4$ used in \cref{thm:locality}. This last inequality is what makes the kernel estimates summable in the resolvent expansion. The radius $\delta_*$ given by~\eqref{eq:delta-star} is sufficient rather than maximal.

\begin{definition}[Admissible stiffness perturbations]
\label{def:admissible-perturbation}
For a displacement $u$, write the Hessian perturbation relative to the homogeneous reference as
\[
  \Hess(u)-\Hhom=\sum_q \mathcal{D}_q^*\,A_q(u)\,\mathcal{D}_q ,
\]
where the sum is over the finite set of Hessian bond legs introduced in \cref{sec:finite-range-structure}. The matrix-valued coefficient $A_q(u)$ is the local bond-stiffness perturbation for leg $q$, namely the relevant second-derivative block of the site potential at displacement $u$ minus its value at the homogeneous reference. The sharp locality estimate below is not claimed for arbitrary site-to-site oscillatory stiffness changes. We call $u$ \emph{admissible} if, uniformly over all bond legs $q$, the coefficients $A_q(u)$ have bounded Schur and $\ell^2$ norms and either are supported in a fixed compact core or satisfy the dyadic multiplier estimate
\begin{equation}
  \|K_i A_q(u) K_j'\|_{\ell^1\to\ell^\infty}
  \leq C_{\rm ad}\,\|A_q(u)\|_{\rm Schur}\,
  2^{-c_{\rm ad}|i-j|}\,2^{-d\max\{i,j\}},
  \label{eq:dyadic-admissibility}
\end{equation}
for all dyadic shell pieces $K_i,K_j'$ of the regularized homogeneous kernels appearing in \cref{lem:chain-algebra}. The constants $C_{\rm ad},c_{\rm ad}>0$ are independent of $i,j$, the supercell size, and the chain length. This condition is a compact way of saying that the stiffness perturbation does not destroy the cancellation of the homogeneous Green-function shells. Compact-core perturbations with fixed core radius and sufficiently regular slowly varying coefficient fields satisfying the shell properties used below are checkable subclasses.
\end{definition}

\begin{proposition}[Sufficient admissibility criteria]
\label{prop:checkable-admissibility}
The admissible class in \cref{def:admissible-perturbation} contains the two cases used as model perturbations in this paper. The compact-core case covers localized stiffness changes. The slowly varying case covers smooth strain or far-field coefficient variations whose length scale is not collapsed to the lattice spacing. The second criterion uses the standard shell properties of the regularized homogeneous kernels: zero moment, the $L^1$/$L^\infty$ annular envelope of \eqref{eq:dyadic-cz}, and the Holder commutator estimate generated by the displayed slow-variation bound.
\begin{enumerate}
  \item If every bond-stiffness perturbation $A_q(u)$ is supported in a defect core of radius $R_{\rm core}$ independent of the supercell size and satisfies uniform Schur and $\ell^2$ bounds, then $u$ is admissible. In this compact-core branch the chain estimate of \cref{lem:chain-algebra} is obtained from finite-dimensional core summation and the endpoint decay of the homogeneous kernels, with constants depending on $R_{\rm core}$ but not on the periodic cell.
  \item If $A_q(u)$ is a multiplication kernel whose entries extend from $\Lambda$ to uniformly bounded $C^{1,\alpha}$ functions on $\mathbb{R}^d$ and satisfy the dyadic slow-variation bound
  \begin{equation}
    \sup_{p\in\Lambda}\sup_{0<|h|\leq 2^j}
    \frac{\|A_q(u;p+h)-A_q(u;p)\|}{\|A_q(u)\|_{\rm Schur}}
    \leq C_{\rm sv}\,2^{\alpha j}L_{\rm sv}^{-\alpha},
    \qquad 2^j\leq L_{\rm sv},
    \label{eq:slow-variation}
  \end{equation}
  with a slow-variation length $L_{\rm sv}$ bounded below by the stencil scale, and with the analogous large-scale bound obtained by the Schur norm for $2^j>L_{\rm sv}$, then the standard shell argument summarized below gives \eqref{eq:dyadic-admissibility}, with constants depending on $C_{\rm sv},\alpha$, the stencil, and the homogeneous multiplier constants. In particular, smooth elastic coefficients whose variation length is not collapsed to the lattice scale fall in the sharp theorem after excluding any non-admissible core into the compact part.
\end{enumerate}
Both criteria are stable under convex interpolation of the displacement field, provided the same core radius or the same slow-variation constants are retained along the path.
\end{proposition}

\emph{Idea of proof.} In the compact-core case all coefficient insertions in a chain are restricted to a fixed finite set, so the two endpoint kernels give
\[
  (1+\dist(\ell,{\rm core}))^{-d}(1+\dist(m,{\rm core}))^{-d}
  \lesssim (1+|\ell-m|)^{-d},
\]
and the finite core matrix powers contribute only a geometric factor controlled by \eqref{eq:kernel-smallness}. For the slowly varying case, write the regularized homogeneous kernels in dyadic shells. If $|i-j|\leq2$, the Schur bound and the shell envelope give
\[
  \|K_iAK'_j\|_{\ell^1\to\ell^\infty}
  \lesssim \|A\|_{\rm Schur}\,2^{-d\max\{i,j\}},
\]
which is \eqref{eq:dyadic-admissibility} without a scale-separation gain. If $i\ll j$, the shell $K'_j$ varies on the larger spatial scale. Use the zero moment of $K_i$ and write the coefficient in the form $A(y)=A(x)+[A(y)-A(x)]$ inside the convolution. The constant part cancels, and the slow-variation estimate gives a factor comparable to $(2^i/2^j)^\alpha$ after summing the annular tails of $K_i$ and $K'_j$; this is $2^{-c|i-j|}$ with $c\leq\alpha$. The amplitude of the broader shell supplies the factor $2^{-d\max\{i,j\}}$. The case $i\gg j$ is identical after exchanging the two kernels and using the zero moment of $K'_j$. These three cases prove \eqref{eq:dyadic-admissibility}. Convex interpolation preserves support and preserves the displayed slow-variation seminorms by convexity of the coefficient differences.

\section{Site Entropy on Multilattices}
\label{sec:site-entropy}

This section turns the finite-cell harmonic entropy into the site-resolved quantity used later as the target of a local surrogate. We begin with the periodic supercell calculation, because this is how entropy labels are produced in atomistic simulations. We then pass to finite-cell limits to obtain supercell-size-independent site labels where that limit is justified. The main estimates show why these labels can be learned from local environments: distant displacements have an algebraically small effect on a site's entropy contribution, and truncating the environment gives a controlled cutoff error. The section ends by translating these estimates into the error decomposition used by the surrogate model.

\subsection{Finite-Cell Site Entropy}
\label{sec:finite-cell}

We first work on a finite periodic supercell, where the harmonic vibrational entropy is the logarithm of a determinant of a finite-dimensional symmetric positive operator and is therefore well defined without auxiliary regularization. The infinite-lattice version is introduced below as the limit of these finite-cell quantities.

Fix a periodic supercell $\Lambda_N:=\Lambda/L_N\Lambda$ of side $L_N$ (equivalently, the quotient torus $\Lambda/L_N\mathbb{Z}^d$ when $\Lambda=A\mathbb{Z}^d$), and write $\mathcal{M}_N:=\{(\ell,\alpha):\ell\in\Lambda_N,\,\alpha\in\Scal\}$ for its multilattice extension. Denote by $\Hess_N(u)$ the finite-cell harmonic Hessian induced by the same site potential as on $\Lambda$ but periodised over $\Lambda_N$, and by $\Hhom_N:=\Hess_N(0)$ its homogeneous reference. The kernel of $\Hhom_N$ has a $d$-dimensional null space corresponding to rigid translations $u_\alpha(\ell)\equiv\mathrm{const}$; let $P_N$ denote the orthogonal projection onto its complement, the \emph{stable subspace} of $\Hhom_N$. All determinants and inverse square roots below are understood on $\operatorname{ran}P_N$, and $\det_+$ denotes the determinant on this subspace (equivalently, the product of the strictly positive eigenvalues of $\Hhom_N$). The finite-cell homogeneous normalization is
\begin{equation}
  \Fhom_N := (\Hhom_N)^{-1/2}\quad\text{on }\operatorname{ran}P_N,
  \qquad
  \Fhom_N:=0\quad\text{on }\ker P_N,
\end{equation}
extended by zero on the rigid-translation modes so that $\Fhom_N$ is a bounded symmetric operator on the full space. The relative harmonic entropy in the cell is
\begin{equation}
  S_N(u)
  :=
  -\frac12\log\det_+\!\left(\Fhom_N^*\Hess_N(u)\Fhom_N\right),
  \label{eq:finite-cell-entropy}
\end{equation}
which assigns zero entropy to the homogeneous reference: $\Fhom_N^*\Hhom_N\Fhom_N=\Id$ on $\operatorname{ran}P_N$ by construction, so $S_N(0)=0$. The finite-cell site contribution is the corresponding diagonal entry of the matrix logarithm,
\begin{equation}
  S^N_{\ell,\alpha}(u)
  :=
  -\frac12\Tr\!\left[\log_+\!\left(\Fhom_N^*\Hess_N(u)\Fhom_N\right)\right]_{(\ell,\alpha),(\ell,\alpha)},
  \qquad(\ell,\alpha)\in\mathcal{M}_N,
  \label{eq:finite-cell-site-entropy}
\end{equation}
where $\log_+(A):=\log(A)P_+$ is the principal-branch logarithm composed with the projection $P_+$ onto the strictly positive part of the spectrum of $A$. The decomposition $S_N(u)=\sum_{(\ell,\alpha)\in\mathcal{M}_N}S^N_{\ell,\alpha}(u)$ holds by construction.
Here and below, the bracketed block is the Cartesian $d\times d$ matrix at the site $(\ell,\alpha)$, and $\Tr$ in \eqref{eq:finite-cell-site-entropy} denotes the trace over those Cartesian components, not a second trace over the full supercell.
Thus $S^N_{\ell,\alpha}(u)$ is the site label associated with the basis atom $(\ell,\alpha)$ in the finite supercell. These labels are the quantities fitted by the local surrogate in \cref{sec:ace}: their definition uses a global matrix logarithm, but the locality estimate below explains why a truncated local environment can carry the relevant information.

\subsection{Finite-Cell Limits and Infinite-Volume Site Labels}
\label{sec:infinite-definition}

The finite-cell definition above also provides the route to a supercell-size-independent site label. We define the infinite-volume label by finite-cell approximation rather than by postulating a diagonal entry of an infinite matrix logarithm. Let $u_N$ be a sequence of periodic representatives that agree with $u$ on balls whose radii tend to infinity. For compact-core admissible $u$ in the stability neighbourhood, \cref{prop:thermodynamic-limit} below proves that the limit
\begin{equation}
  S_{\ell,\alpha}(u)
  :=
  \lim_{N\to\infty}S^N_{\ell,\alpha}(u_N)
  \label{eq:Sla-infinite}
\end{equation}
exists for each fixed site $(\ell,\alpha)$ and is independent of the approximating sequence. For non-compact dyadically admissible coefficient fields, the locality estimates remain pointwise kernel estimates once the required periodic kernel/resolvent convergence is available, but the present thermodynamic-limit construction is fully proved only in the compact-core regime. General elastic far-fields are therefore used below through the explicit convolution bound in \cref{lem:localized-perturbation}, not as an unconditional sharp-rate infinite-volume entropy theorem.

Let $\Hess(u)$ be the harmonic Hessian on the infinite multilattice $\Lambda\times\Scal$ at displacement $u\in\Uh$, identified with a bounded symmetric operator on $\ell^2(\Lambda\times\Scal;\mathbb{R}^d)$ as in \cref{sec:setting}, and let $\Hhom:=\Hess(0)$. On the infinite $\ell^2$ space the constant rigid translations are not square-summable modes; the relevant zero-mode removal is therefore understood as the thermodynamic limit of the finite-cell projections $P_N$ in \cref{sec:finite-cell}. Equivalently, in the Fourier representation of the translation-invariant operator $\Hhom$, we remove the acoustic point $\xi=0$ before taking the inverse square root. This point has measure zero in the infinite-lattice multiplier, but its removal records the same normalization as the finite periodic determinant and makes the finite-cell approximations uniform.

By \cref{ass:stability}, $\widehat\Hhom(\xi)$ is positive definite for $\xi\neq0$, with the acoustic eigenvalues behaving like $|\xi|^2$ near the origin and the optical eigenvalues bounded below by $\gamma_{\rm p}$. We define
\begin{equation}
  \widehat\Fhom(\xi):=\widehat\Hhom(\xi)^{-1/2}\quad(\xi\in\BZ\setminus\{0\}),
  \qquad
  \widehat\Fhom(0):=0
  \quad\text{in finite-cell regularizations},
  \label{eq:Fhom-def}
\end{equation}
and let $\Fhom$ be the corresponding positive self-adjoint Fourier multiplier. It is unbounded on $\ell^2$ because of the acoustic singularity $|\xi|^{-1}$ described in \eqref{eq:Fhat-block}, but the difference-regularized objects $\mathcal{D}\Fhom$ used below are bounded by \cref{lem:hom-kernel}. The formal infinite-volume normalized Hessian is the quadratic form
\begin{equation}
  T(u) := \Fhom^* \Hess(u) \Fhom,
  \label{eq:T-def}
\end{equation}
defined a priori on $\operatorname{dom}\Fhom$. At the homogeneous reference, $T(0)=\Id$ on the zero-frequency-regularized sector. \Cref{lem:localized-perturbation} below shows that for every admissible $u\in\mathcal{B}_{\delta_*}$ the perturbation $T(u)-\Id$ admits the bounded representation $\sum_q(\mathcal{D}_q\Fhom)^*A_q(u)(\mathcal{D}_q\Fhom)$; consequently $T(u)$ extends to a bounded symmetric operator with $\|T(u)-\Id\|\leq1/2$. After the decay estimates have been proved, the finite-cell limit \eqref{eq:Sla-infinite} is equivalently represented by
\begin{equation}
  S_{\ell,\alpha}(u)
  =
  -\frac12
  \Tr\!\left[\log_+\!\left(T(u)\right)\right]_{(\ell,\alpha),(\ell,\alpha)}
  \label{eq:Sla-log-kernel}
\end{equation}
whenever the right-hand side is interpreted as the limit of zero-frequency-regularized finite-cell kernels. The total vibrational entropy is the formal sum
\begin{equation}
  S(u) = \sum_{\ell\in\Lambda}\sum_{\alpha\in\Scal} S_{\ell,\alpha}(u),
\end{equation}
which converges absolutely in the compact-core regime. For non-compact far fields we work pointwise with $S_{\ell,\alpha}$ and use the far-field convolution bound only as a finite-cell consistency check.

\subsection{Locality Mechanism}
\label{sec:locality}

The estimates below explain why the finite-cell site labels can be treated as local quantities for regression. In a multilattice, this locality is less direct than in a Bravais lattice because the homogeneous Hessian contains acoustic modes, optical internal-shift modes, and coupling between them. The lemmas in this subsection record the kernel estimates needed to propagate this block structure through the homogeneous normalization and the perturbed resolvent. These estimates feed into the locality and truncation results stated in \cref{sec:locality-thm}; full proof details are collected in \cref{app:proofs}.

\subsubsection{Why Multilattice Locality Is Nontrivial}
\label{sec:proof-architecture}

The proof follows the Bravais-lattice entropy-locality strategy of~\cite{TorabiWangOrtner2024Entropy}, but the physical content of a multilattice changes the analysis. The homogeneous Hessian contains acoustic translations, optical internal shifts, and mixed acoustic--optical couplings; the entropy label is built from a normalized Hessian and a matrix logarithm; and the defect enters through local stiffness changes. The estimates below keep these ingredients in the operator order needed by the entropy derivative.

\begin{remark}[What changes from the Bravais case]
\label{rem:new-vs-bravais}
The Bravais ingredients that transfer directly are the contour representation of $\log_+$, differentiation under the contour integral, the finite-range support of $\partial\Hess(u)$, and the final convolution that turns two rate-$d$ factors into the entropy-gradient rate. The multilattice-specific work is concentrated in three mechanisms:
\begin{enumerate}
  \item The homogeneous symbol has an acoustic singularity and optical internal-shift modes. The inverse square root $\widehat\Fhom(\xi)$ therefore has different orders in the acoustic, mixed, and optical blocks, rather than a single scalar order.
  \item The finite-difference legs from the Hessian regularize the acoustic singularity. After this regularization, the kernels used by the entropy derivative have the uniform real-space decay needed for local regression.
  \item The perturbed resolvent must be estimated in the non-commuted order that appears in the differentiated entropy formula. The block resolvent estimate transfers the homogeneous acoustic--optical decay hierarchy through this resolvent without changing the final algebraic rate.
\end{enumerate}
A more detailed layer-by-layer comparison is collected in \cref{tab:bravais-vs-multi} of \cref{app:proofs}. The optical gap, full block-coercivity margin, acoustic--optical coupling bounds, and admissibility constants enter the prefactors, while the algebraic exponent is preserved for the admissible perturbation class.
\end{remark}

\subsubsection{Supporting Kernel Estimates}
\label{sec:block-lemmas}

The supporting estimates have distinct roles in the computational argument. \Cref{lem:block-sqrt-symbol,lem:hom-kernel} show that the homogeneous normalization has cutoff-compatible kernels after the Hessian difference legs are applied. \Cref{lem:block-conv,lem:chain-algebra} explain why the borderline kernel products must be handled with cancellation rather than by absolute convolution alone. \Cref{lem:localized-perturbation,lem:block-resolvent} then pass these estimates from the homogeneous reference to the perturbed finite-cell Hessian. Readers focused on the surrogate calculation can view these lemmas as the analytical support for \cref{thm:locality}; the detailed estimates are given in \cref{app:proofs}.

\begin{lemma}[Block square-root multiplier bounds]
\label{lem:block-sqrt-symbol}
Under \cref{ass:smoothness,ass:stability}, the zero-frequency-regularized multiplier $\widehat\Fhom(\xi)=\widehat\Hhom(\xi)^{-1/2}$ satisfies the acoustic/optical block expansion
\begin{equation}
  \widehat\Fhom(\xi)
  =
  \begin{pmatrix}
    |\xi|^{-1}B_{00}(\hat\xi)+O(1) & B_{0p}(\hat\xi)+O(|\xi|)\\
    B_{p0}(\hat\xi)+O(|\xi|) & B_{pp}(\hat\xi)+O(|\xi|)
  \end{pmatrix},
  \qquad \hat\xi:=\xi/|\xi|,
  \label{eq:block-sqrt-expansion}
\end{equation}
on $\BZ\setminus\{0\}$, where the angular matrices are smooth on the unit sphere away from the usual coordinate-chart seams and their norms are controlled by the full stability margin, the optical gap, and the acoustic--optical coupling bounds. For every multi-index $\rho$,
\begin{equation}
  |\partial_\xi^\rho \widehat F_{00}(\xi)|
  \leq C_\rho|\xi|^{-1-|\rho|},
  \qquad
  |\partial_\xi^\rho \widehat F_{0p}(\xi)|
  +|\partial_\xi^\rho \widehat F_{p0}(\xi)|
  +|\partial_\xi^\rho \widehat F_{pp}(\xi)|
  \leq C_\rho|\xi|^{-|\rho|},
  \label{eq:sqrt-block-derivatives}
\end{equation}
for all $\xi\neq0$. Consequently, for every Hessian leg $\mathcal D$ and every pair of Hessian legs $\mathcal D,\mathcal D'$,
\begin{equation}
  \left|\partial_\xi^\rho\big(\widehat{\mathcal D}(\xi)\widehat\Fhom(\xi)^*\big)\right|
  +
  \left|\partial_\xi^\rho\big(\widehat{\mathcal D}(\xi)\widehat\Hhom(\xi)^{-1}\widehat{\mathcal D'}(\xi)^*\big)\right|
  \leq C_\rho |\xi|^{-|\rho|},
  \qquad \xi\neq0.
  \label{eq:block-marcinkiewicz}
\end{equation}
The same estimates hold for the adjoint products. The constants are uniform for the finite periodic Brillouin grids after removing the zero-frequency acoustic mode.
\end{lemma}

\emph{Idea of proof.} Write $\varepsilon:=|\xi|$ and express the homogeneous symbol in the acoustic/optical basis as
\[
  \widehat\Hhom(\xi)=
  \begin{pmatrix}
    \varepsilon^2 A_{00}(\hat\xi)+O(\varepsilon^4) & \varepsilon A_{0p}(\hat\xi)+O(\varepsilon^2)\\
    \varepsilon A_{p0}(\hat\xi)+O(\varepsilon^2) & A_{pp}(0)+O(\varepsilon)
  \end{pmatrix}.
\]
The full block coercivity in \cref{ass:stability} and the optical gap make the optical block uniformly invertible and make the acoustic Schur complement comparable to $\varepsilon^2$. Hence, for $\varepsilon$ small, the spectrum splits into an acoustic cluster of $d$ eigenvalues in $[c\varepsilon^2,C\varepsilon^2]$ and an optical cluster in $[c,C]$, with constants uniform in the direction $\hat\xi$. Let $\Pi_a(\xi)$ be the Riesz spectral projector onto the acoustic cluster and $\Pi_p(\xi)=I-\Pi_a(\xi)$. The block perturbation is off-diagonal of size $O(\varepsilon)$ relative to the optical gap, so Kato perturbation theory gives
\[
  \Pi_a(\xi)=P_0+O(\varepsilon),\qquad
  P_0\Pi_a(\xi)P_p+P_p\Pi_a(\xi)P_0=O(\varepsilon),
\]
and, after differentiating the Riesz integral for the projector,
$|\partial_\xi^\rho\Pi_a(\xi)|\leq C_\rho\varepsilon^{-|\rho|}$ and
$|\partial_\xi^\rho(\Pi_a(\xi)-P_0)|\leq C_\rho\varepsilon^{1-|\rho|}$ for the orders used here. On the acoustic cluster,
\[
  \widehat\Hhom(\xi)^{-1/2}\Pi_a(\xi)
  =
  \varepsilon^{-1}\Pi_a(\xi)\,B_a(\hat\xi,\varepsilon)\,\Pi_a(\xi),
\]
where $B_a$ is uniformly bounded, uniformly elliptic on $\operatorname{ran}\Pi_a(\xi)$, and satisfies the angular Marcinkiewicz bounds obtained by applying the holomorphic functional calculus to the rescaled acoustic operator $\varepsilon^{-2}\widehat\Hhom(\xi)|_{\operatorname{ran}\Pi_a(\xi)}$. On the optical cluster,
\[
  \widehat\Hhom(\xi)^{-1/2}\Pi_p(\xi)
  =
  \Pi_p(\xi)\,B_p(\hat\xi,\varepsilon)\,\Pi_p(\xi),
\]
with $B_p$ uniformly bounded and satisfying the same derivative pattern without the prefactor $\varepsilon^{-1}$. These two spectral-calculus representations, not a congruence rule for matrix square roots, give the block orders: the acoustic--acoustic block has size $O(\varepsilon^{-1})$, the mixed blocks have size $\varepsilon^{-1}O(\varepsilon)+O(1)=O(1)$ because $\Pi_a-P_0=O(\varepsilon)$, and the optical block is $O(1)$. Differentiating the Riesz and functional-calculus integrals gives \eqref{eq:sqrt-block-derivatives}. Multiplication by a Hessian leg contributes an $O(|\xi|)$ factor on the acoustic projector and a bounded factor on the optical projector, so the products in \eqref{eq:block-marcinkiewicz} are degree-zero Marcinkiewicz multipliers. This is the square-root version of the block Green hierarchy used in \cite[Lem.~3.5 and Thm.~3.6]{OlsonOrtnerWangZhang2023MultilatticeDislocations}, stated here for the exact products needed below.

\begin{lemma}[Homogeneous normalized kernel]
\label{lem:hom-kernel}
Let $\mathcal{D}$ denote any difference or shift leg generated by the finite-range multilattice Hessian stencil. Under \cref{ass:smoothness,ass:stability}, the homogeneous kernel satisfies
\begin{equation}
  \left|
  \left[\mathcal{D}\Fhom^*\right]_{(\ell,\alpha),(m,\gamma)}
  \right|
  +
  \left|
  \left[\Fhom\mathcal{D}^*\right]_{(\ell,\alpha),(m,\gamma)}
  \right|
  \lesssim
  (1+|\ell-m|)^{-d}.
  \label{eq:hom-kernel}
\end{equation}
The same estimate holds uniformly for the periodic finite-cell operators away from the translational zero modes.
\end{lemma}

\emph{Idea of proof.} The lattice Fourier transform reduces the kernel to an oscillatory integral of the symbol $\sigma(\xi)=\widehat{\mathcal{D}}(\xi)\widehat\Fhom(\xi)^*$ on the Brillouin zone. The block decomposition of $\widehat\Fhom$ near $\xi=0$,
\begin{equation}
  \widehat\Fhom(\xi)
  =
  \begin{pmatrix}
    O(|\xi|^{-1}) & O(1) \\
    O(1) & O(1)
  \end{pmatrix},
  \qquad \xi\to 0,
  \label{eq:Fhat-block}
\end{equation}
exhibits the acoustic-block singularity $|\xi|^{-1}$ formalized in \cref{lem:block-sqrt-symbol}. By the projector calculation in \cref{sec:setting}, every Hessian leg has an $O(|\xi|)$ restriction on the acoustic subspace and is uniformly bounded on the optical subspace; multiplying by $\widehat{\mathcal{D}}(\xi)$ therefore cancels the acoustic pole and leaves $\sigma$ bounded with conical singularity at the origin. The kernel decay $r^{-d}$ then follows from a near/far split at scale $\eta=1/r$ combined with $d+1$ integrations by parts on the far piece; one extra IBP beyond the conventional $d$ closes the dyadic geometric series without a logarithm. The full argument is given in \cref{app:hom-kernel}.

The convolution estimate that appears repeatedly below is the multilattice analogue of the scalar lattice convolution lemma of \cite[Lemma~3.3]{TorabiWangOrtner2024Entropy}. We state the absolute-value version explicitly because the borderline rate is important: composing two kernels that both decay exactly like $(1+r)^{-d}$ produces a logarithmic loss unless an additional cancellation estimate is used.

\begin{lemma}[Block lattice convolution]
\label{lem:block-conv}
Let $K_1,K_2:(\Lambda\times\Scal)^2\to\mathbb{R}^{d\times d}$ be operator kernels and assume there exist exponents $a_1,a_2>0$ with $a_1+a_2>d$ such that
\begin{equation}
  \big|[K_i]_{(\ell,\alpha),(m,\gamma)}\big|\leq C_i\,(1+|\ell-m|)^{-a_i}
  \quad\text{uniformly in }\alpha,\gamma\in\Scal,\ i=1,2.
  \label{eq:block-conv-input}
\end{equation}
Set
\[
  \mu(a_1,a_2):=\min\{a_1,a_2,a_1+a_2-d\},
  \qquad
  \vartheta(a_1,a_2):=\mathbf{1}_{\{a_1=d\}}+\mathbf{1}_{\{a_2=d\}} .
\]
Then the composition $K_1K_2$ satisfies, uniformly in $\alpha,\gamma\in\Scal$,
\begin{equation}
  \big|[K_1K_2]_{(\ell,\alpha),(m,\gamma)}\big|
  \leq c_d\,|\Scal|\,C_1 C_2\,
  (1+|\ell-m|)^{-\mu(a_1,a_2)}
  \log^{\vartheta(a_1,a_2)}(2+|\ell-m|)
  \label{eq:block-conv-output}
\end{equation}
where $c_d$ is a dimension-dependent constant inherited from the scalar lattice convolution estimate \cite[Lemma~3.3]{TorabiWangOrtner2024Entropy}. Thus neither the critical product $d$ with $d$ nor the endpoint pairing $d+1$ with $d-1$ is sufficient, by absolute values alone, to close the rate-$d$ resolvent chain; the sharper rate in \cref{lem:block-resolvent} uses the finite-difference cancellation before taking absolute values.
\end{lemma}

\emph{Idea of proof.} The block product is a finite sublattice sum of scalar lattice convolutions, giving only the prefactor $|\Scal|$; see \cref{app:block-conv}. The point of the lemma is negative as much as positive: repeated rate-$d$ products cannot be closed sharply by absolute values alone, so the resolvent estimate below must use cancellation before the final absolute-value bound is taken.

\begin{lemma}[Cancellation-sensitive chain algebra]
\label{lem:chain-algebra}
Let $G^{\rm hom}:=\Fhom\Fhom^*=(\Hhom)^{-1}$ on the zero-frequency-regularized sector. Fix a finite sequence of Hessian legs $q_1,\ldots,q_k$ and multiplication kernels $A_{q_j}$ that are admissible in the sense of \cref{def:admissible-perturbation}, with Schur and $\ell^2$ norms bounded by $A$. For the homogeneous factors
\[
  L_{q_1}:=\mathcal{D}G^{\rm hom}\mathcal{D}_{q_1}^*,\qquad
  K_{q_jq_{j+1}}:=\mathcal{D}_{q_j}G^{\rm hom}\mathcal{D}_{q_{j+1}}^*,\qquad
  R_{q_k}:=\mathcal{D}_{q_k}\Fhom,
\]
where $j=1,\ldots,k-1$, the chain kernel satisfies, for every $k\geq1$,
\begin{equation}
  \left|
  \left[
  L_{q_1}A_{q_1}
  \Big(\prod_{j=1}^{k-1}K_{q_jq_{j+1}}A_{q_{j+1}}\Big)
  R_{q_k}
  \right]_{(\ell,\alpha),(m,\gamma)}
  \right|
  \leq
  C_{\rm ch}\,(C_{\rm hom}A)^k\,(1+|\ell-m|)^{-d}.
  \label{eq:chain-algebra}
\end{equation}
The constants depend only on the dimension, the finite stencil, the stability constants, $|\Scal|$, and the admissibility constants of the coefficients. In the compact-core branch of \cref{def:admissible-perturbation} they may also depend on the fixed core radius, but not on $k$, $\ell,m$, the supercell size, or the particular coefficient values.
\end{lemma}

\emph{Idea of proof.} The compact-core case follows from fixed finite support and the two endpoint kernels. In the non-compact branch, the homogeneous kernels are decomposed into zero-mean dyadic shells, and the coefficient factors are used only through the admissibility estimate \eqref{eq:dyadic-admissibility}. The shell interaction gains the summable factor $2^{-c|i-j|}$, so the scale summation does not produce a logarithm or a polynomial factor in the chain length. This restricted coefficient class is precisely the setting in which the sharp rate is asserted; outside it, absolute convolution generally gives a weaker bound.

\begin{lemma}[Localized normalized perturbation]
\label{lem:localized-perturbation}
Let $u\in\mathcal{B}_{\delta_*}$ (the stability neighbourhood of \cref{def:stab-nbhd}) and let $\delta\Hess(u):=\Hess(u)-\Hhom$ denote the perturbation of the Hessian relative to the homogeneous reference. The normalized operator $T(u)$ defined in \eqref{eq:T-def} admits the representation
\begin{equation}
  T(u)-\Id
  =
  \sum_q
  \Fhom^*\mathcal{D}_q^*\,A_q(u)\,\mathcal{D}_q\Fhom,
  \label{eq:localized-perturbation}
\end{equation}
where the sum runs over the finite set of bond legs of the multilattice stencil. Here and below $A_q(u)$ denotes the \emph{bond-stiffness perturbation} of leg $q$, i.e. the local second-derivative block at displacement $u$ minus its homogeneous value at $u=0$, not the full homogeneous bond stiffness. If $\delta\Hess(u)$ is supported in a finite neighbourhood of a defect core (compact-core regime), then the kernel of $T(u)-\Id$ satisfies
\begin{equation}
  \left|\left[T(u)-\Id\right]_{(\ell,\alpha),(m,\gamma)}\right|
  \lesssim
  (1+\dist(\ell,\operatorname{core}))^{-d}\,
  (1+\dist(m,\operatorname{core}))^{-d}.
  \label{eq:localized-perturbation-decay}
\end{equation}
For an elastic far field the corresponding bound is the explicit convolution estimate
\begin{equation}
  \left|\left[T(u)-\Id\right]_{(\ell,\alpha),(m,\gamma)}\right|
  \lesssim
  \sum_{p\in\Lambda}a(p)
  (1+|\ell-p|)^{-d}(1+|m-p|)^{-d},
  \qquad
  a(p):=\max_q\|A_q(u;p)\|.
  \label{eq:farfield-conv}
\end{equation}
In particular, if $a(p)\lesssim(1+\dist(p,\operatorname{core}))^{-(d-1+\beta)}$, the diagonal far-field regime is controlled by the diagonal-dominated rate of this convolution rather than by the product of two independent core distances.
\end{lemma}

\emph{Idea of proof.} The bond-leg representation of the Hessian perturbation gives \eqref{eq:localized-perturbation} after homogeneous normalization. Compact-core coefficients are summed over a fixed finite set, while non-compact far fields retain the explicit convolution \eqref{eq:farfield-conv}. The latter keeps the diagonal-dominated contribution visible, which is why the product core-localization bound is not asserted for elastic far-fields.

\begin{lemma}[Block normalized resolvent estimate]
\label{lem:block-resolvent}
Let $C$ be a contour enclosing the spectrum of $T(u)=\Fhom^*\Hess(u)\Fhom$ uniformly away from zero, and let $R_z=(z\Id-T(u))^{-1}$ for $z\in C$. Suppose $u\in\mathcal{B}_{\delta_*}$ is admissible in the sense of \cref{def:admissible-perturbation}, and that the stability radius has been chosen so that the kernel-majorant smallness \eqref{eq:kernel-smallness} holds. For every finite-difference leg $\mathcal{D}$ generated by the finite-range Hessian stencil, the block kernels needed in the differentiated entropy formula satisfy the non-commuted estimate
\begin{equation}
  \left|
  \left[(\mathcal{D}\Fhom)R_z\right]_{(\ell,\alpha),(m,\gamma)}
  \right|
  +
  \left|
  \left[R_z(\Fhom^*\mathcal{D}^*)\right]_{(\ell,\alpha),(m,\gamma)}
  \right|
  \lesssim
  (1+|\ell-m|)^{-d},
  \label{eq:block-resolvent}
\end{equation}
uniformly in $z\in C$ and in the finite periodic approximations. The hidden constant depends on the stability margin, the interaction range, the dimension, the number of sublattices, the admissibility constants of $A_q(u)$, and the margin $1-\theta$ in \eqref{eq:kernel-smallness}.
\end{lemma}

\emph{Idea of proof.} Expand the resolvent by a Neumann series around the identity, using the smallness condition in \cref{def:stab-nbhd}. After inserting the bond-leg representation of $T(u)-\Id$, each term becomes a chain of the type estimated in \cref{lem:chain-algebra}, with the regularized homogeneous kernels at the endpoints and in the interior. The adjoint chain gives the second factor in \eqref{eq:block-resolvent}. No commutation of the resolvent with the difference legs is used.

\begin{remark}
The important point in \cref{lem:block-resolvent} is the operator order. The derivative of $T(u)$ in \cref{eq:dT-leg-decomp} places the two regularized legs as $R_z(\mathcal{D}_q\Fhom)^*$ and $(\mathcal{D}_q\Fhom)R_z$; no commutation of $\mathcal{D}_q$, $\Fhom$, and the perturbed resolvent is used. The finite-difference regularization is carried by the homogeneous factors $\mathcal{D}G^{\rm hom}\mathcal{D}_q^*$, $\mathcal{D}_qG^{\rm hom}\mathcal{D}_{q'}^*$, and $\mathcal{D}_q\Fhom$, and the sharp rate is preserved by the admissible dyadic chain algebra rather than by an almost-orthogonality statement for arbitrary rough multipliers.
\end{remark}

\subsection{Locality and Truncation Estimates}
\label{sec:locality-thm}

The kernel estimates now translate into the computational statement needed by a local surrogate: changing atoms far from a site has a controlled, algebraically decaying effect on that site's entropy label. This is the estimate that justifies fitting $S^N_{\ell,\alpha}$, and its finite-cell limit, from a cutoff environment rather than from the whole supercell.

\begin{theorem}[Multilattice site-entropy locality]
\label{thm:locality}
Suppose \cref{ass:smoothness,ass:stability} hold and $u\in\mathcal{B}_{\delta_*}$ is a compact-core admissible displacement in the sense of \cref{def:admissible-perturbation}. Let $S_{\ell,\alpha}(u)$ be the finite-cell-limit site entropy constructed in \cref{prop:thermodynamic-limit}. Then, for all $\ell,n \in \Lambda$ and all $\alpha,\beta \in \mathcal{S}$,
\begin{equation}
  \left|
  \frac{\partial S_{\ell,\alpha}(u)}
       {\partial u_\beta(n)}
  \right|
  \leq
  C_{\rm loc}\left(1 + |\ell-n|\right)^{-2d},
  \label{eq:locality-bound}
\end{equation}
with constant $C_{\rm loc}$ independent of $\ell,n$ and the supercell size. The constant depends on the full stability margin, the optical gap, acoustic--optical coupling and block-projector conditioning, the stencil diameter, $|\Scal|$, $d$, the admissibility constants of $u$, and $\|\partial^{(\leq 3)}V\|$.

For dyadically admissible non-compact coefficient fields satisfying the same smallness assumptions, the finite-cell derivative bound of \cref{cor:finite-cell-locality} holds uniformly. The corresponding infinite-volume pointwise derivative estimate is asserted only when the periodic kernel/resolvent convergence described in \cref{sec:infinite-definition} is available; in that case $S_{\ell,\alpha}$ is interpreted through that conditional finite-cell limit.
\end{theorem}

The locality constant $C_{\rm loc}$ depends on derivatives of the site potential up to order three: the second derivative enters through $\Hhom$ and the third through $\partial_{u_\beta(n)}\Hess(u)$. The fourth-order norm $\|\partial^{(\leq 4)}V\|$ does not appear in $C_{\rm loc}$; it enters only through the radius $\delta_*$ of the stability neighbourhood in \eqref{eq:delta-star}, where it controls a Taylor-remainder bound on $\|T(u)-\Id\|$ uniform on $\mathcal{B}_{\delta_*}$.

\begin{proof}
We give the finite-cell proof skeleton; the kernel estimates that make the argument uniform in the cell size are isolated in the preceding lemmas and in \cref{app:proofs}. Work first on a finite periodic cell, where every operator is a finite matrix on $\operatorname{ran}P_N$. The compact-core thermodynamic-limit statement follows in \cref{prop:thermodynamic-limit}, while the non-compact statement is only the conditional kernel-limit statement described in the theorem. We suppress the subscript $N$ in the notation. Let $P_{\ell,\alpha}$ denote the orthogonal projection onto the displacement components at the site $(\ell,\alpha)$. By \cref{def:stab-nbhd} the spectrum of $T(u)$ on the zero-frequency-regularized sector lies in $[1/2,3/2]$. We take $C$ to be the positively oriented circle of radius $r_0=3/4$ centered at $z=1$, which encloses this spectrum uniformly and stays at distance $\geq 1/4$ from both the spectrum and from $z=0$. The principal-branch matrix logarithm admits the contour representation
\begin{equation}
  \log_+ T(u) = \frac{1}{2\pi i}\int_C \log(z)\,(z\Id-T(u))^{-1}\dd z,
  \label{eq:contour-log}
\end{equation}
valid because $T(u)$ is symmetric positive definite on the zero-frequency-regularized sector (\cref{ass:stability} and the smallness of $u$) and $C$ encloses the spectrum without crossing the cut $(-\infty,0]$. Substituting into the finite-cell definition gives the finite-cell identity, and the same display represents the compact-core limit after \cref{prop:thermodynamic-limit}; in the notation of the limit,
\begin{equation}
  S_{\ell,\alpha}(u)
  =
  -\frac{1}{2}\Tr\!\left[P_{\ell,\alpha}\log_+ T(u)\right]
  =
  -\frac{1}{4\pi i}\int_C \log(z)\,\Tr\!\left[P_{\ell,\alpha}\,R_z\right]\dd z,
\end{equation}
where $R_z=(z\Id-T(u))^{-1}$ is the resolvent. Differentiation under the integral, justified by the uniform boundedness of $R_z$ on the contour, yields
\begin{equation}
  \frac{\partial S_{\ell,\alpha}(u)}{\partial u_\beta(n)}
  =
  -\frac{1}{4\pi i}\int_C \log(z)\,\Tr\!\left[
    P_{\ell,\alpha}\,R_z\,(\partial_{u_\beta(n)}T(u))\,R_z
  \right]\dd z.
  \label{eq:dS-contour}
\end{equation}

The derivative of the normalized operator is
\begin{equation}
  \partial_{u_\beta(n)}T(u)
  =\Fhom^*\,\partial_{u_\beta(n)}\Hess(u)\,\Fhom
  =\sum_{q\in B(n)}(\mathcal{D}_q\Fhom)^*\,\partial_{u_\beta(n)}A_q(u)\,(\mathcal{D}_q\Fhom),
  \label{eq:dT-leg-decomp}
\end{equation}
where $B(n)$ is the finite set of bond legs of the multilattice stencil that touch site $n$. The matrix $M_q(n,\beta):=\partial_{u_\beta(n)}A_q(u)$ is supported on a fixed-radius neighbourhood of $n$, uniformly bounded by $\|\partial^{(\leq 3)}V\|$, and acts at site indices $(p,\sigma),(p',\sigma')\in B_{r_V}(n)\times\Scal$.

Inserting \eqref{eq:dT-leg-decomp} into \eqref{eq:dS-contour} and using $\Tr\!\big[P_{\ell,\alpha}\,X\,P_{\ell,\alpha}\big]_{(\ell,\alpha),(\ell,\alpha)} = X_{(\ell,\alpha),(\ell,\alpha)}$ for any operator $X$,
\begin{align}
  \frac{\partial S_{\ell,\alpha}(u)}{\partial u_\beta(n)}
  ={}&-\frac{1}{4\pi i}\sum_{q\in B(n)}\sum_{\substack{(p,\sigma)\\(p',\sigma')\in B_{r_V}(n)}}
  \int_C\log(z)\,\big[R_z(\Fhom^*\mathcal{D}_q^*)\big]_{(\ell,\alpha),(p,\sigma)}\,
  [M_q(n,\beta)]_{(p,\sigma),(p',\sigma')} \notag\\
  &\hspace{6em}\times\big[(\mathcal{D}_q\Fhom)R_z\big]_{(p',\sigma'),(\ell,\alpha)}\,\dd z.
  \label{eq:dS-explicit}
\end{align}
Bounding the inner integrand by the supremum over the contour and applying \cref{lem:block-resolvent} to each of the two block kernels at the displaced site indices $(p,\sigma),(p',\sigma')\in B_{r_V}(n)$,
\begin{equation}
  \big|[R_z(\Fhom^*\mathcal{D}_q^*)]_{(\ell,\alpha),(p,\sigma)}\big|
  +\big|[(\mathcal{D}_q\Fhom)R_z]_{(p',\sigma'),(\ell,\alpha)}\big|
  \lesssim (1+|\ell-n|)^{-d},
  \qquad z\in C,
\end{equation}
where the implicit constant depends on $r_V$ and absorbs the $(1+|p-\ell|)^{-d}\sim(1+|\ell-n|)^{-d}$ shift valid for $|\ell-n|\geq 2 r_V$ and trivial otherwise. Multiplying the two factors, summing the finitely many bond legs $q\in B(n)$, the finite numbers of inner site indices $(p,\sigma),(p',\sigma')$, and bounding $|M_q(n,\beta)|$ by $\|\partial^{(\leq 3)}V\|$,
\begin{equation}
  \left|\frac{\partial S_{\ell,\alpha}(u)}{\partial u_\beta(n)}\right|
  \leq
  \frac{|C|}{2\pi}\,\sup_{z\in C}|\log z|\,
  \big|B(n)\big|\,(2r_V+1)^{2d}|\Scal|^2\,
  \|\partial^{(\leq 3)}V\|\,(1+|\ell-n|)^{-2d},
\end{equation}
which gives \eqref{eq:locality-bound} with the constant dependencies stated in the theorem.

The multilattice-specific input is \cref{lem:block-resolvent}: the acoustic, mixed, and optical Green-function rates of \cite[Theorem~3.6]{OlsonOrtnerWangZhang2023MultilatticeDislocations} are all converted, after homogeneous normalization and one Hessian difference leg, into the common rate $(1+|\ell|)^{-d}$. The two non-commuted resolvent factors in \eqref{eq:dS-explicit} therefore give the product rate in \eqref{eq:locality-bound}.
\end{proof}

\begin{remark}
The block resolvent estimate of \cref{lem:block-resolvent} is what allows a single cutoff-local entropy map to remain valid in the presence of acoustic--optical coupling. It transfers the multilattice Green-function hierarchy through the normalized resolvent in the same operator order that appears in the differentiated trace formula, preserving the Bravais entropy-gradient decay rate for the admissible perturbation class.
\end{remark}

\begin{corollary}[Uniform finite-cell locality]
\label{cor:finite-cell-locality}
Under the finite-cell version of the assumptions used in \cref{thm:locality}, with either compact-core admissibility or dyadic admissibility and kernel-majorant smallness on the periodic cell, the finite-cell site entropy satisfies
\begin{equation}
  \left|
  \frac{\partial S^N_{\ell,\alpha}(u)}
       {\partial u_\beta(n)}
  \right|
  \leq
  C_{\rm loc}\left(1+\operatorname{dist}_N(\ell,n)\right)^{-2d}
  \;+\;C_{\rm img}\,L_N^{-2d},
\end{equation}
where $\operatorname{dist}_N$ is the periodic lattice distance, $L_N$ is the minimum distance from the defect core to its nearest nontrivial periodic image when such a compact core is present, and $C_{\rm img}$ has the same dependencies as $C_{\rm loc}$. In particular, the bound is uniform in $N$ whenever $L_N\geq c_*\,(1+\operatorname{dist}_N(\ell,n))$ for any fixed $c_*>0$, that is, when the nearest periodic image of the defect is at least at the same distance as the site $n$ being differentiated. The image term is a uniform periodic-cell control sufficient for the thermodynamic-limit argument; it is not intended as a sharp finite-size asymptotic.
\end{corollary}

\subsubsection{Truncation theorem}
\label{sec:truncation-thm}

The next result converts locality into a cutoff rule. It bounds the error made when the site label is evaluated from a truncated displacement field, and therefore supplies the deterministic cutoff contribution in the surrogate error estimate.

\begin{theorem}[Truncation of multilattice site entropy]
\label{thm:truncation}
Let
\begin{equation}
  B_{\rcut}(\ell) := \{(n,\beta) : |n-\ell| \leq \rcut,\ \beta \in \mathcal{S}\}.
\end{equation}
Define the truncated displacement field by
\[
  u^{(\rcut)}_\beta(n):=
  \begin{cases}
  u_\beta(n), & (n,\beta)\in B_{\rcut}(\ell),\\
  0, & (n,\beta)\notin B_{\rcut}(\ell),
  \end{cases}
\]
and set $\widetilde S_{\ell,\alpha}(u;\rcut):=S_{\ell,\alpha}(u^{(\rcut)})$.
Assume that the hard-cutoff interpolation between $u^{(\rcut)}$ and $u$ remains in the same admissible class with uniform constants. This is automatic for compact-core perturbations once the cutoff contains the core, but it is not automatic for slowly varying tails because a hard cutoff may create an artificial jump at the cutoff boundary. Then
\begin{equation}
  \left|S_{\ell,\alpha}(u) - \widetilde S_{\ell,\alpha}(u;\rcut)\right|
  \leq
  C_{\rm tr}\,\rcut^{-d}\,\|u\|_{L^\infty(\Lambda\setminus B_{\rcut}(\ell))},
  \label{eq:truncation-bound}
\end{equation}
with constant $C_{\rm tr}$ that has the same dependencies as $C_{\rm loc}$ in \cref{thm:locality} and is independent of the supercell size.
\end{theorem}

\begin{proof}
Define the linear interpolation $u^{(s)}=(1-s)u^{(\rcut)}+s\,u$ for $s\in[0,1]$. By construction, $u^{(0)}=u^{(\rcut)}$ and $u^{(1)}=u$, so
\begin{equation}
  S_{\ell,\alpha}(u)-\widetilde S_{\ell,\alpha}(u;\rcut)
  =\int_0^1\!\sum_{n\notin B_{\rcut}(\ell)}\sum_{\beta\in\Scal}
  \frac{\partial S_{\ell,\alpha}(u^{(s)})}{\partial u_\beta(n)}\,
  u_\beta(n)\dd s.
  \label{eq:fund-thm-calc}
\end{equation}
On the right-hand side, $u^{(s)}$ remains in the stability neighborhood because $\|Du^{(s)}\|_{\ell^\infty}\leq\|Du\|_{\ell^\infty}$ and the admissibility constants are uniform along the path, so \cref{thm:locality} applies to each integrand:
\begin{equation}
  \left|\frac{\partial S_{\ell,\alpha}(u^{(s)})}{\partial u_\beta(n)}\right|
  \leq
  C_{\rm loc}(1+|\ell-n|)^{-2d}.
\end{equation}
Combining with \eqref{eq:fund-thm-calc} and the support condition $|n-\ell|>\rcut$,
\begin{equation}
  \left|S_{\ell,\alpha}(u)-\widetilde S_{\ell,\alpha}(u;\rcut)\right|
  \leq
  C_{\rm loc}\,|\Scal|\,\|u\|_{L^\infty(\Lambda\setminus B_{\rcut}(\ell))}\,
  \sum_{|n-\ell|>\rcut}(1+|n-\ell|)^{-2d}.
  \label{eq:tr-tail}
\end{equation}
The tail sum on the lattice $\Lambda\subset\mathbb{R}^d$ satisfies $\sum_{|n|>\rcut}(1+|n|)^{-2d}\lesssim\rcut^{-d}$, which combined with \eqref{eq:tr-tail} gives \eqref{eq:truncation-bound}.
\end{proof}

\begin{corollary}[Smooth truncation for slowly varying tails]
\label{cor:smooth-truncation}
Let $\chi_{\rcut}:\Lambda\to[0,1]$ be a lattice cutoff with $\chi_{\rcut}=1$ on $B_{\rcut}(\ell)$, $\chi_{\rcut}=0$ outside $B_{2\rcut}(\ell)$, and finite differences satisfying $|D^m\chi_{\rcut}|\leq C_m\rcut^{-m}$ for the stencil orders used by the energy. Define the smoothly truncated field $u^{[\rcut]}_\beta(n):=\chi_{\rcut}(n)u_\beta(n)$ and $\widetilde S^{\rm sm}_{\ell,\alpha}(u;\rcut):=S_{\ell,\alpha}(u^{[\rcut]})$. If the interpolation between $u^{[\rcut]}$ and $u$ is uniformly admissible, then
\begin{equation}
  \left|S_{\ell,\alpha}(u)-\widetilde S^{\rm sm}_{\ell,\alpha}(u;\rcut)\right|
  \leq
  C_{\rm tr}\,\rcut^{-d}\,
  \|u\|_{L^\infty(\Lambda\setminus B_{\rcut}(\ell))}.
  \label{eq:smooth-truncation-bound}
\end{equation}
For the slowly varying subclass of \cref{prop:checkable-admissibility}, the derivative bounds on $\chi_{\rcut}$ preserve the dyadic slow-variation constants up to cutoff-independent factors for $\rcut$ larger than the stencil scale. Thus smooth truncation, rather than hard truncation, is the natural admissible path for algebraic tails.
\end{corollary}

\begin{proof}
The difference $u-u^{[\rcut]}$ is supported outside $B_{\rcut}(\ell)$ and has $L^\infty$ norm bounded by $\|u\|_{L^\infty(\Lambda\setminus B_{\rcut}(\ell))}$. Applying the fundamental-theorem argument in the proof of \cref{thm:truncation} along the admissible smooth-cutoff path gives the same tail sum $\sum_{|n-\ell|>\rcut}(1+|n-\ell|)^{-2d}\lesssim\rcut^{-d}$, which proves \eqref{eq:smooth-truncation-bound}. The final assertion follows from the product rule for finite differences: the new coefficient variation is the sum of the original slow variation and terms containing $D^m\chi_{\rcut}$, each bounded by $C_m\rcut^{-m}$ and therefore uniform at the dyadic scales relevant to the cutoff transition.
\end{proof}

\begin{remark}[Practical choice of $\rcut$]
\label{rem:practical-rcut}
The localization on the right-hand side of \eqref{eq:truncation-bound} and \eqref{eq:smooth-truncation-bound} is the $L^\infty$ norm of the displacement \emph{outside} the truncation radius. For a compact-core model perturbation supported in $B_{R_{\rm core}}(\ell)$, the hard-cutoff bound is exact for $\rcut\geq R_{\rm core}$. For an admissible algebraic tail $|u(n)|\lesssim (1+|n|)^{-\beta}$, $\beta>0$, the smooth-cutoff path gives the sharper practical estimate $|S_{\ell,\alpha}-\widetilde S^{\rm sm}_{\ell,\alpha}|\lesssim \rcut^{-d-\beta}$, provided the admissibility constants stay uniform along the smooth truncation path. A practitioner targeting site RMSE $\eta$ should choose $\rcut\sim (C_{\rm tr}\|u\|_{L^\infty}/\eta)^{1/d}$ as a conservative rule; the constant worsens near optical softening, loss of full-block coercivity, or ill-conditioned acoustic--optical splitting.
\end{remark}

\subsection{Supercell Consistency and Surrogate Error Estimates}
\label{sec:supercell-error-estimates}

The preceding estimates are most useful for training when finite-supercell labels converge to a site label that does not depend on the chosen periodic box. The next lemma and proposition give this consistency for compact-core admissible perturbations. The argument deliberately keeps the homogeneous normalization, resolvent, and matrix logarithm in the comparison; the local bond stencil alone is not enough, because the normalized entropy label remains a nonlocal finite-cell quantity before the limit is taken.

\begin{lemma}[Derivative convergence in the compact-core limit]
\label{lem:derivative-limit}
Let $u$ be a compact-core admissible displacement satisfying the assumptions of \cref{thm:locality}, and let $u_N$ be periodic representatives whose image distance tends to infinity. For any fixed differentiated site $(n,\beta)$,
\begin{equation}
  \partial_{u_\beta(n)}S_{\ell,\alpha}(u)
  =
  \lim_{N\to\infty}
  \partial_{u_\beta(n)}S^N_{\ell,\alpha}(u_N),
  \label{eq:derivative-limit}
\end{equation}
where the derivative on the left is the derivative of the finite-cell-limit map along compactly supported perturbations.
\end{lemma}

\begin{proof}
Apply the differentiated contour formula \eqref{eq:dS-explicit} in the finite cell. The derivative kernel is a finite sum of products of the two non-commuted resolvent factors in \cref{lem:block-resolvent} and the compactly supported matrix $\partial_{u_\beta(n)}A_q(u)$. Uniform locality gives a summable dominating tail independent of $N$, while the zero-mode-regularized homogeneous kernels and contour resolvents converge on each fixed local block. Dominated convergence therefore permits the limit of the finite-cell derivatives to be identified with the derivative of the finite-cell-limit site entropy along compactly supported perturbations.
\end{proof}

\begin{proposition}[Supercell consistency of site labels]
\label{prop:thermodynamic-limit}
Let $u$ be a compact-core admissible displacement satisfying the assumptions of \cref{thm:locality}, and let $u_N$ be periodic representatives whose image distance $L_N$ tends to infinity. For each fixed $(\ell,\alpha)$, the finite-cell values $S^N_{\ell,\alpha}(u_N)$ form a Cauchy sequence. Hence the limit in \eqref{eq:Sla-infinite} exists and is independent of the approximating sequence. More quantitatively, for any cutoff $\rcut<L_N/4$,
\begin{equation}
  \left|S^N_{\ell,\alpha}(u_N)-S_{\ell,\alpha}(u)\right|
  \leq
  C\,\rcut^{-d}\,\|u\|_{L^\infty}
  +\varepsilon_N(\rcut),
  \label{eq:thermodynamic-limit}
\end{equation}
where $\varepsilon_N(\rcut)\to0$ for each fixed $\rcut$ as $N\to\infty$. If a periodic-image expansion for the homogeneous finite-cell kernels is used, then $\varepsilon_N(\rcut)$ is bounded by the corresponding image tail; in the present argument only its convergence to zero is needed.
The derivative convergence along compactly supported perturbations is stated separately in \cref{lem:derivative-limit}.
\end{proposition}

\begin{proof}
The proof compares finite cells only through objects whose local kernels have already been controlled. Fix $\rcut<L_N/4$ and split
\begin{equation}
  \left|S^N_{\ell,\alpha}(u_N)-S_{\ell,\alpha}(u)\right|
  \leq
  \left|S^N_{\ell,\alpha}(u_N)-\widetilde S^N_{\ell,\alpha}(u_N;\rcut)\right|
  +\left|\widetilde S^N_{\ell,\alpha}(u_N;\rcut)-\widetilde S_{\ell,\alpha}(u;\rcut)\right|
  +\left|\widetilde S_{\ell,\alpha}(u;\rcut)-S_{\ell,\alpha}(u)\right|.
\end{equation}
The first and third terms are bounded by the finite-cell and infinite-cell truncation estimates, uniformly in $N$. The middle term is not set to zero: even when the local bond stencil does not wrap around the periodic boundary, the homogeneous normalizations $\Fhom_N$, the resolvents, and the matrix logarithms remain nonlocal. Instead it is bounded by $\varepsilon_N(\rcut)$, the finite-cell convergence error of the zero-mode-regularized homogeneous kernels and contour resolvents restricted to $B_{\rcut}(\ell)$. For fixed $\rcut$, this error tends to zero by the periodic Riemann-sum convergence of the Fourier multipliers away from the removed acoustic point and by the uniform resolvent bound on the contour. Taking $N\to\infty$ for fixed $\rcut$ and then $\rcut\to\infty$ gives the Cauchy property and \eqref{eq:thermodynamic-limit}.

The argument gives convergence only for fixed cutoff before the limit $\rcut\to\infty$ is taken. The present theorem therefore does not claim an explicit finite-size convergence rate for $\varepsilon_N(\rcut)$; any image-tail rate must be proved separately from the chosen periodic-kernel expansion.
\end{proof}

\subsubsection{Surrogate error estimates}
\label{sec:error-decomp}

The final step is to express what the locality and truncation estimates buy in a fitted model. \Cref{thm:locality,thm:truncation} control the deterministic error of the truncated target $\widetilde S_{\ell,\alpha}$. The actual error of a fitted surrogate carries two further contributions: a basis approximation error that depends on the regularity of $\widetilde S_{\ell,\alpha}$ relative to the model class, and an estimation error that depends on the finite training set. To make this practical decomposition explicit, write the fitted surrogate as the Tikhonov-regularized empirical-risk minimizer
\begin{equation}
  S^*_{\ell,\alpha}(\cdot;\hat c) = \mathop{\arg\min}_{f\in\mathcal{H}_{\nu,p}}
  \frac{1}{M}\sum_{m=1}^{M}\bigl|f(\mathcal{E}_m)-\widetilde S_{\ell_m,\alpha_m}(u_m;\rcut)\bigr|^2
  +\lambda\,\|f\|_{\mathcal{H}_{\nu,p}}^2,
\end{equation}
where $\mathcal{H}_{\nu,p}\subset L^2$ is a finite-dimensional parametric site-model class indexed by body order $\nu$ and total degree $p$ (the multi-species ACE basis of \cref{sec:ace}), $\mathcal{E}_m$ is the truncated environment of the $m$-th training site, and $\lambda>0$ is a Tikhonov regularization parameter.

\begin{remark}[Surrogate error decomposition and sample complexity]
\label{rem:error-decomp}
Let $\mathcal{H}_{\nu,p}\subset L^2(\mu)$ be a fixed parametric site-model class with effective dimension $\dim\mathcal{H}_{\nu,p}$ (in particular, the multi-species ACE basis of body order $\nu$ and total degree $p$ used in \cref{sec:ace}, for which $\dim\mathcal{H}_{\nu,p}=O(C_{\Scal,\alpha}p^\nu)$ at fixed radial/angular channel family and fixed sublattice partition; see \cite{Drautz2019ACE,Bachmayr2022ACECompleteness}). Assume that the training environments are iid from $\mu$, the feature vectors and labels are uniformly bounded, and $\hat c$ is the Tikhonov-regularized empirical-risk minimizer with fixed $\lambda>0$. This sample-complexity statement is an iid-environment idealization; it is not a proof for correlated site labels extracted from the same supercell, for which configuration-level splits or block bootstrap checks are more appropriate. The constant $C_{\Scal,\alpha}$ accounts for the number of species and central-sublattice coefficient blocks. Combining \cref{thm:truncation} with the standard Rademacher-complexity bound for ridge regression on bounded-norm regressors \cite{Cucker2002LearningTheory,Wendland2004Scattered}, with probability at least $1-\delta$,
\begin{equation}
  \big\|S_{\ell,\alpha}-S^*_{\ell,\alpha}(\cdot;\hat c)\big\|_{L^2(\mu)}
  \leq
  \underbrace{C_{\rm tr}\,\rcut^{-d}\,\|u\|_{L^\infty}}_{\varepsilon_{\rm trunc}}
  +\underbrace{\inf_{f\in\mathcal{H}_{\nu,p}}\big\|\widetilde S_{\ell,\alpha}-f\big\|_{L^2(\mu)}}_{\varepsilon_{\rm approx}}
  +\underbrace{C_{\rm est}\!\left(\sqrt{\tfrac{\dim\mathcal{H}_{\nu,p}}{M}}+\sqrt{\tfrac{\log(1/\delta)}{M}}\right)}_{\varepsilon_{\rm est}}.
  \label{eq:error-decomposition}
\end{equation}
The first term is the genuinely multilattice contribution: \cref{thm:truncation} couples the cutoff radius $\rcut$ to the dimension $d$ in a way that is independent of basis order, species count, and training-set size. Optimising \eqref{eq:error-decomposition} for a target accuracy $\eta$ on the ACE class above yields the sample-complexity scaling
\begin{equation}
  \rcut\sim\eta^{-1/d},\qquad p\sim\eta^{-1/s},\qquad M\sim\eta^{-(\nu/s+2)},
  \label{eq:sample-complexity}
\end{equation}
where $s>0$ is the regularity exponent of the truncated site target $\widetilde S_{\ell,\alpha}(\cdot;\rcut)$ as a function of the local geometry. The finite-smoothness hypothesis in \cref{ass:smoothness} by itself gives only a finite regularity exponent, limited by the differentiability of the site potential and by the distance to the boundary of the stability neighbourhood $\mathcal{B}_{\delta_*}$. In the special case of analytic or polynomial test potentials, the finite-dimensional truncated Hessian depends analytically on the displacement as long as the spectrum stays in the stable sector; for the synthetic spring models of \cref{sec:numerics:locality,sec:numerics:truncation} this idealized setting gives $s=\infty$ and \eqref{eq:sample-complexity} reduces to the algebraic rate $M\sim\eta^{-2}$. For Stillinger--Weber and similar three-body potentials the angular cutoff function is only $C^k$ for finite $k$ on the sampled displacement geometry, so the effective $s$ in \eqref{eq:sample-complexity} is finite and produces a strictly larger formal sample-complexity exponent than in the analytic case. The cutoff exponent in $\rcut\sim\eta^{-1/d}$ inherits the dimension of the locality theorem; the constants in $C_{\rm tr}$ and $C_{\rm est}$ deteriorate with the stability margins, optical gap, block coupling, feature normalization, and ridge parameter, but the scaling exponents do not.
\end{remark}

\section{Multispecies Local Surrogates}
\label{sec:ace}

The theory above leads to a direct surrogate calculation. For a set of training supercells, the expensive step is performed once: assemble the finite-cell Hessian, diagonalize it on the positive-mode sector, and extract the site labels $S^N_{\ell,\alpha}$. Around each labelled site we then build a cutoff-local environment, fit a species- and sublattice-resolved site map, and deploy the fitted model on new supercells by summing the predicted site contributions. The purpose of this section is to specify this local representation and the corresponding training/deployment cost.

\subsection{Cutoff-Local Site Representation}
\label{sec:ace-environment}

For a site $(\ell,\alpha)$ and a cutoff radius $\rcut$, define the local environment
\begin{equation}
  \mathcal E_{\ell,\alpha}^{\rcut}(y)
  :=
  \left\{
  \big(y_\beta(n)-y_\alpha(\ell),\,\beta\big):
  |n-\ell|\leq \rcut,\ \beta\in\Scal
  \right\}.
  \label{eq:local-environment}
\end{equation}
The relative position removes global translations, while the neighbour label $\beta$ records the chemical species and crystallographic sublattice channel carried by the multilattice basis. A site model for harmonic entropy is a map
\[
  \mathcal E_{\ell,\alpha}^{\rcut}(y)\longmapsto S^*_{\ell,\alpha}(y;c),
\]
with the central label $\alpha$ selecting the appropriate coefficient block. The locality and truncation estimates of \cref{thm:locality,thm:truncation} provide the analytical justification for replacing the global entropy label by this cutoff-local input.

We use a species-resolved ACE basis~\cite{Drautz2019ACE,Bachmayr2022ACECompleteness,Bachmayr2020Symmetric}:
\begin{equation}
  S^*_{\ell,\alpha}(y;c)
  =
  \sum_j c_{\alpha j}\,
  B_{\alpha j}\!\left(\mathcal E_{\ell,\alpha}^{\rcut}(y)\right),
  \label{eq:ace-site}
\end{equation}
where $B_{\alpha j}$ is a rotation-invariant and species-wise permutation-invariant ACE basis function for central channel $\alpha$. The coefficient vector $c_{\alpha j}$ is shared across crystallographically equivalent sublattices, but split by sublattice label when the reference multilattice carries inequivalent sites. This distinction matters even when there is only one chemical element: in diamond Si, zinc-blende compounds, and B2 alloys, chemically identical or chemically paired atoms may occupy different basis positions and couple to different internal-shift modes. Suppressing the sublattice label would fold these environments together and lose part of the point-group and basis symmetry. For crystallographically equivalent basis sites, the coefficient block is shared and \eqref{eq:ace-site} reduces to the usual species-only ACE representation.

The total surrogate is the extensive sum
\begin{equation}
  S^*(y;c)=\sum_{\ell,\alpha}S^*_{\ell,\alpha}(y;c).
  \label{eq:ace-total}
\end{equation}
Thus the same fitted site map can be evaluated on any compatible supercell, provided the local environments remain within the geometric and stability range represented in the training data.

\subsection{Training, Cost, and Deployment}
\label{sec:ace-training}

The primary regression target is the site entropy label. For training sites indexed by $m=1,\ldots,M_{\rm site}$, with environments $\mathcal E_m$ and labels $S_m^{\rm site}$, the site-only ridge problem is
\begin{equation}
  \hat c
  =
  \mathop{\arg\min}_{c}
  \frac{1}{M_{\rm site}}\sum_{m=1}^{M_{\rm site}}
  \left|S_m^{\rm site}-S^*(\mathcal E_m;c)\right|^2
  +\lambda\|c\|^2 .
  \label{eq:ace-site-loss}
\end{equation}
Here $S^*(\mathcal E_m;c)$ denotes the site model \eqref{eq:ace-site} evaluated on the $m$-th local environment with its central channel. This is the loss used in the numerical experiments of \cref{sec:numerics}. It directly targets the object controlled by \cref{thm:locality}, while the total entropy is recovered by the sum \eqref{eq:ace-total}.

When total entropy or entropy-gradient labels are available, they can be added as auxiliary consistency information:
\begin{equation}
  \mathcal{L}(c)
  =
  w_{\rm site}\mathcal L_{\rm site}(c)
  +w_{\rm tot}\mathcal L_{\rm tot}(c)
  +w_{\rm grad}\mathcal L_{\rm grad}(c),
  \label{eq:ace-loss}
\end{equation}
where $\mathcal L_{\rm site}$ is \eqref{eq:ace-site-loss}, while $\mathcal L_{\rm tot}$ and $\mathcal L_{\rm grad}$ are the corresponding total-entropy and gradient losses. The present experiments use the site-only choice $w_{\rm site}=1$ and $w_{\rm tot}=w_{\rm grad}=0$; total-entropy and gradient blocks are left to future training-set design.

The computational gain is an amortization of the global harmonic-entropy calculation. Let $N_{\rm at}$ denote the number of atoms in a supercell. Reference-label generation requires one Hessian assembly and one dense eigendecomposition of a $dN_{\rm at}\times dN_{\rm at}$ matrix per training configuration, with cost $O((dN_{\rm at})^3)$ for the diagonalization. In our finite-difference implementation, Hessian assembly requires a number of force calls proportional to $dN_{\rm at}$, and each force call is $O(N_{\rm at})$ for finite-range potentials. Exact site-label generation through the spectral representation of $\log_+T$ remains a global matrix-function operation after this eigendecomposition.

After the labels have been generated, basis assembly is linear in the number of training environments and basis functions, $O(M_{\rm site}P)$ for $P=\dim\mathcal H_{\nu,p}$, whereas a direct dense ridge solve costs $O(M_{\rm site}P^2+P^3)$ unless an iterative or structured solver is used. Once trained, the surrogate evaluation is local: each site sees $O(\rcut^d)$ neighbours at fixed density, so evaluating and summing \eqref{eq:ace-site} over a supercell costs $O(N_{\rm at}\rcut^d)$ at fixed basis size and cutoff. The numerical scaling tests in \cref{sec:numerics:scaling} compare this deployment cost with direct dense diagonalisation.

\FloatBarrier

\section{Numerical Experiments}
\label{sec:numerics}

The numerical study is organized around the computational use of the theory. Controlled multilattice models first test whether site-entropy gradients and cutoff errors follow the rate envelopes needed for local regression. Finite-range Stillinger--Weber Si and CdTe benchmarks then test whether sublattice and species labels improve entropy-site prediction. Finally, scaling and transfer tests measure the practical gain after the expensive Hessian labels have been generated once. In all Hessian-label calculations the logarithm is evaluated on the finite-cell positive-mode sector of \cref{sec:finite-cell}; the experiments test the calculation strategy and rate behaviour, but do not attempt to infer the full set of stability or admissibility constants.

\subsection{Controlled Rate and Cutoff Tests}
\label{sec:numerics:locality}
\label{sec:numerics:truncation}

The first set of experiments isolates the analytical rate mechanism before moving to empirical potentials. We use three finite-range spring multilattices: a diatomic chain, a two-sublattice square lattice, and a CsCl-type three-dimensional lattice. Each model contains acoustic and optical branches, and each allows the optical separation to be reduced by weakening selected same-sublattice or intercell springs. This provides a controlled way to check whether the decay slope is stable while the prefactor worsens near optical softening.

\Cref{fig:locality} reports the envelope of $|\partial S^N_{\ell_0,\alpha}/\partial u_\beta(n)|$ as a function of distance from the reference site. The 1D and 2D tests show clear algebraic regimes matching the $(1+r)^{-2d}$ reference slopes predicted by \cref{thm:locality}. In both cases, reducing the optical separation raises the prefactor without changing the observed slope, which is the behaviour expected from the stability-dependent constants. The 3D CsCl-type test is limited by the $5^3$ dense-diagonalisation supercell; within that accessible range the envelope is consistent with the $(1+r)^{-6}$ reference after first-shell discretization effects are excluded.

The second part of the same controlled test asks whether the locality estimate translates into a usable cutoff rule. We compare the full-cell site entropy $S^N_{0,0}(u)$ with the truncated value $\widetilde S_{0,0}(u;\rcut)$, using periodic cells large enough to separate the cutoff error from the finite-cell image error. \Cref{fig:truncation} shows the resulting cutoff curves. The 1D and 2D localized perturbations decay faster than the conservative $\rcut^{-d}$ envelope because their displacement tails also decay. The 3D cosine modulation keeps $\|u\|_{L^\infty(\Lambda\setminus B_{\rcut})}$ essentially constant in $\rcut$ and therefore gives the closest comparison to the worst-case $\rcut^{-3}$ bound. These tests support the use of $\rcut$ as a controllable accuracy parameter for the local surrogate.

\begin{figure}[t]
\centering
\begin{subfigure}[b]{0.32\textwidth}
\centering
\includegraphics[width=\linewidth]{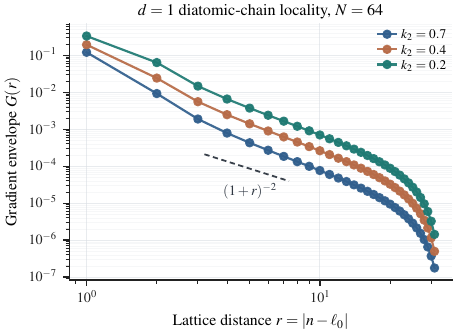}
\caption{$d=1$: diatomic chain.}
\end{subfigure}\hfill
\begin{subfigure}[b]{0.32\textwidth}
\centering
\includegraphics[width=\linewidth]{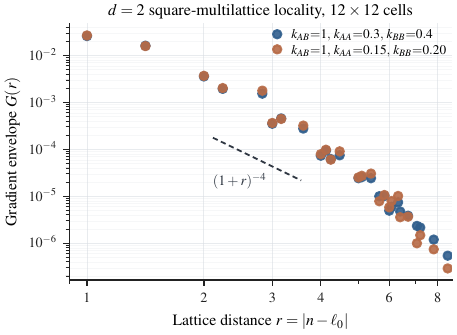}
\caption{$d=2$: square multilattice.}
\end{subfigure}\hfill
\begin{subfigure}[b]{0.32\textwidth}
\centering
\includegraphics[width=\linewidth]{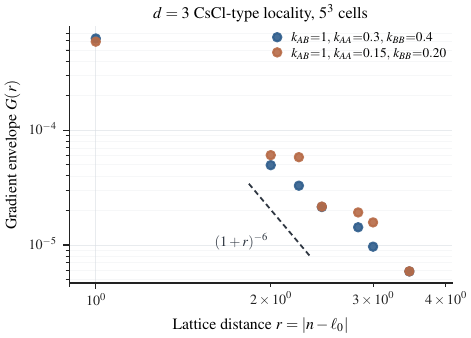}
\caption{$d=3$: CsCl-type multilattice.}
\end{subfigure}
\caption{Envelope of $|\partial S^N_{\ell_0,\alpha}/\partial u_\beta(n)|$ as a function of $|n-\ell_0|$ in $d=1,2,3$. The 1D and 2D panels show the median over independent random base configurations, with the 10\%--90\% interquantile band shaded; the 3D panel shows a single finite-difference run due to cost. Dashed reference slopes correspond to $(1+r)^{-2d}$ in each panel.}
\label{fig:locality}
\end{figure}

\begin{figure}[t]
\centering
\begin{subfigure}[b]{0.32\textwidth}
\centering
\includegraphics[width=\linewidth]{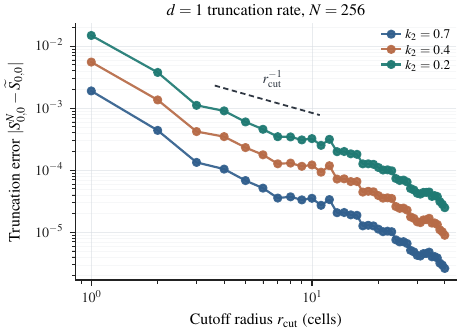}
\caption{$d=1$: diatomic chain.}
\end{subfigure}\hfill
\begin{subfigure}[b]{0.32\textwidth}
\centering
\includegraphics[width=\linewidth]{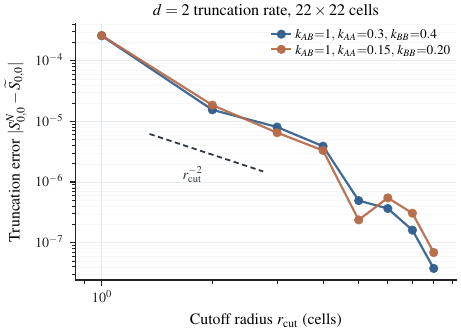}
\caption{$d=2$: square multilattice.}
\end{subfigure}\hfill
\begin{subfigure}[b]{0.32\textwidth}
\centering
\includegraphics[width=\linewidth]{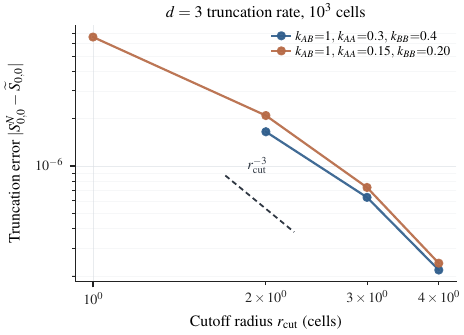}
\caption{$d=3$: CsCl-type multilattice.}
\end{subfigure}
\caption{Truncation error $|S^N_{0,0}(u)-\widetilde S_{0,0}(u;\rcut)|$ as a function of cutoff radius $\rcut$ in $d=1,2,3$. Dashed reference slopes $\rcut^{-d}$. The empirical decay sits below the upper bound; in 1D and 2D the extra decay reflects the localized displacement tails, while the 3D modulation is closer to the worst-case envelope.}
\label{fig:truncation}
\end{figure}

\FloatBarrier

\subsection{Sublattice- and Species-Resolved Entropy Regression}
\label{sec:numerics:si-sw}
\label{sec:numerics:cdte-sw}

The next question is whether the site label can be learned by a local representation that keeps the multilattice labels visible. Diamond silicon under the Stillinger--Weber three-body potential~\cite{StillingerWeber1985Silicon} is a useful first benchmark because it has one chemical element but two inequivalent basis sites coupled by an optical internal-shift mode. The Hessians are computed with ASE~\cite{LarsenASE2017} and LAMMPS~\cite{ThompsonLAMMPS2022}; the benchmark is used as a finite-range, stable small-perturbation test of the surrogate calculation rather than as a reconstruction of all theorem constants.

For Si, we generate $30$ localized perturbations of a $2\times2\times2$ conventional cubic supercell, giving $1920$ site labels. A sublattice-resolved ACE model from \texttt{ACEpotentials.jl}~\cite{Witt2023ACEpotentials,Lysogorskiy2021PerformantImplementationACE} with body order $\nu=3$, total degree $p=8$, cutoff $\rcut=5.5$\,\AA{}, and basis dimension $62$ captures most of the site-entropy variance. The randomized site split measures interpolation across local environments, while the configuration-level split trains on $24$ full configurations and tests on $6$ held-out configurations. Both splits remain accurate, showing that the fit is not only reusing sites from the same distorted cell. The Cartesian baseline in \cref{fig:si-ace} is also sublattice-resolved, so the improvement comes from the richer symmetry-adapted ACE basis rather than from adding the sublattice label alone.

\begin{figure}[t]
\centering
\includegraphics[width=\linewidth]{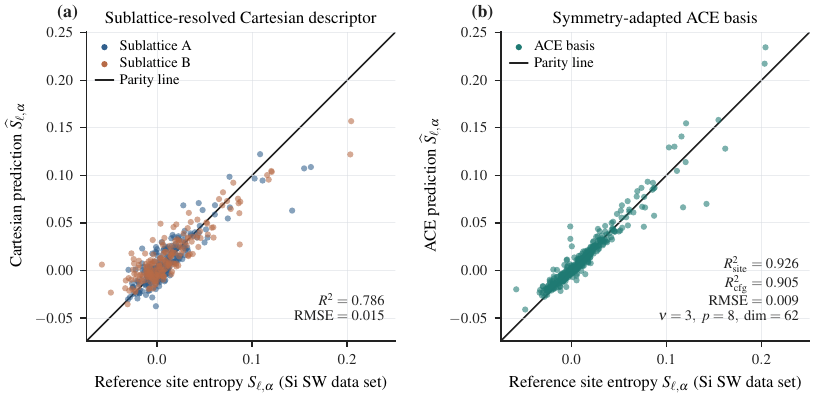}
\caption{Multispecies ACE site model on Si SW. Held-out parity at body order $\nu=3$, total degree $p=8$, cutoff $\rcut=5.5$\,\AA{}; basis dimension $62$. The ACE fit reaches $R^2=0.926$ on a randomized site split and $R^2=0.905$ on a stricter configuration-level split. On the same randomized split, a sublattice-resolved Cartesian descriptor baseline attains $R^2=0.79$, consistent with the approximation term in \cref{rem:error-decomp}.}
\label{fig:si-ace}
\end{figure}

\FloatBarrier

The error scale is consistent with the decomposition in \cref{rem:error-decomp}. At the chosen cutoff the estimated truncation contribution is below $10^{-2}$, and the iid finite-dimensional estimate gives the scale $\sigma_y\sqrt{\dim\mathcal{H}_{\nu,p}/M_{\rm site}}\approx 6\times10^{-3}$ after rescaling by the site-label standard deviation. This is the same order as the randomized-site RMSE and below the configuration-level RMSE reported in \cref{fig:si-ace}. As in \cref{rem:error-decomp}, this estimate is a formal iid guide; the configuration split is the relevant check against correlations among site labels from the same supercell.

CdTe tests the complementary case in which the two sublattices also carry different chemical species. We use zinc-blende CdTe under a finite-range SW-family parameterization as a two-species benchmark, not as a calibrated CdTe thermodynamics model; the parameter source and LAMMPS potential file will be included in the reproducibility archive. A body-order-$2$ species-resolved Cartesian polynomial is fitted to site labels from $24$ localized perturbations of a $2\times2\times2$ conventional cubic supercell. The purpose is a label-ablation test: when Cd and Te are collapsed into one chemical type, both the central-site and neighbour-species channels are removed from the descriptor. \Cref{fig:cdte-sw} shows that this collapse loses a large fraction of the explainable variance. Retaining the species channels restores the accuracy, supporting the multilattice requirement that inequivalent basis sites remain separated in the local entropy map.

\begin{figure}[t]
\centering
\includegraphics[width=0.95\linewidth]{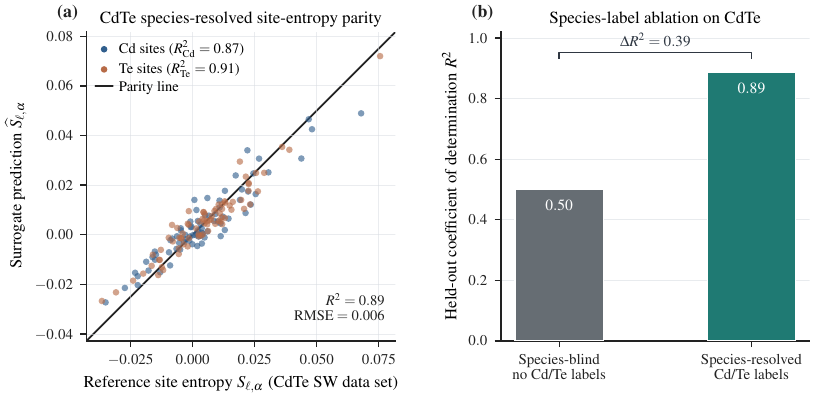}
\caption{Species-resolved local-polynomial test on CdTe SW. (a) Held-out site parity for a body-order-$2$ Cartesian polynomial fit with species channels, with per-sublattice $R^2_{\rm Cd}=0.87$, $R^2_{\rm Te}=0.91$ (overall $R^2=0.89$). (b) Species-label ablation on the same held-out split: collapsing Cd and Te into a single chemical type reduces the overall $R^2$ to about $0.50$, whereas retaining the species-resolved channels gives $R^2=0.89$ ($\Delta R^2\approx0.39$).}
\label{fig:cdte-sw}
\end{figure}

\FloatBarrier

\subsection{Linear-Scaling Evaluation and Supercell Transfer}
\label{sec:numerics:scaling}

\Cref{sec:ace} separates a one-time label-generation cost from repeated local evaluation. The Hessian labels require dense spectral calculations on the training configurations, but a fitted site map evaluates with fixed cutoff and can be summed over a new supercell. We test this deployment regime in two ways: a one-dimensional timing sweep where direct diagonalisation is available over a broad size range, and a three-dimensional Si test where the trained local map is evaluated on larger cells.

\Cref{fig:scaling} reports the one-dimensional timing sweep for $N\in\{32,\dots,2048\}$ cells. Across this factor-$64$ increase in system size, direct evaluation grows by approximately $5\times10^4$, while surrogate evaluation grows by approximately $65$, following the expected linear envelope. The direct curve is below the ideal cubic envelope at smaller sizes because finite-difference Hessian assembly contributes before dense eigendecomposition dominates. At fixed $\rcut=4$, the total-entropy error remains uniformly below about $3\times10^{-3}$ across the sweep, consistent with the cutoff estimate.

\Cref{fig:si-scaling-3d}(b) compares measured full direct Si entropy evaluation with the trained local ACE site model on the displayed supercell sizes $N\in\{216,512,1000,1728,2744\}$ atoms. The direct timings include finite-difference Hessian assembly and the normalized log-determinant evaluation. The local timings include ACE site-descriptor evaluation and the fitted ridge map at fixed cutoff, with model construction and compilation excluded by warm-up. Across the displayed range, the full direct cost grows from $5.90$\,s to $593$\,s, while the local ACE evaluation grows from $5.0\times10^{-3}$\,s to $6.2\times10^{-2}$\,s. At $N=1000$ the trained site map is about $2.8\times10^3$ times faster per entropy evaluation, and at $N=2744$ the measured speedup is about $9.6\times10^3$.

\Cref{fig:si-scaling-3d}(a) tests whether the learned ACE site map transfers across supercell size. The model trained on the $2\times2\times2$ Si data set is evaluated without retraining on an independent $3\times3\times3$ test set generated from the same localized-perturbation distribution. On $940$ held-out site labels from $20$ configurations, the transfer accuracy remains close to the within-size split, with comparable performance on the two sublattices. This supports the intended calculation: dense Hessian diagonalisation supplies reusable site labels for training, after which the local map can be deployed on larger supercells at fixed cutoff.

\begin{figure}[t]
\centering
\includegraphics[width=\linewidth]{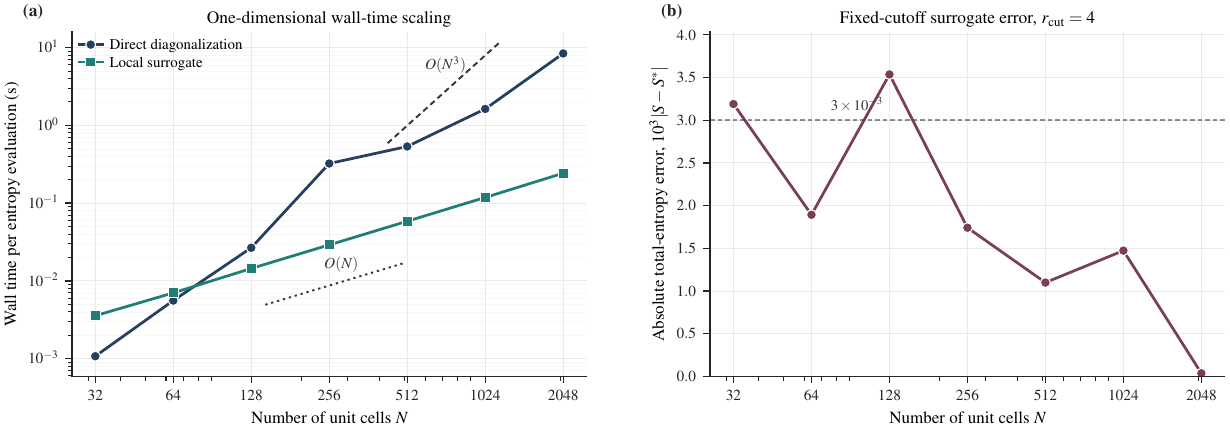}
\caption{(a) Wall-time scaling of direct diagonalization versus surrogate evaluation on the 1D diatomic chain for $N\in\{32,\dots,2048\}$ cells. Reference slopes $O(N^3)$ and $O(N)$ are shown; the direct curve approaches the cubic envelope only in the upper part of the range, with an $O(N^2)$ Hessian-assembly contribution dominating at small $N$. (b) Total-entropy error of the surrogate versus cell size at fixed $\rcut=4$, bounded uniformly by $\approx 3\times 10^{-3}$ in agreement with \cref{thm:truncation}.}
\label{fig:scaling}
\end{figure}

\begin{figure}[t]
\centering
\includegraphics[width=\linewidth]{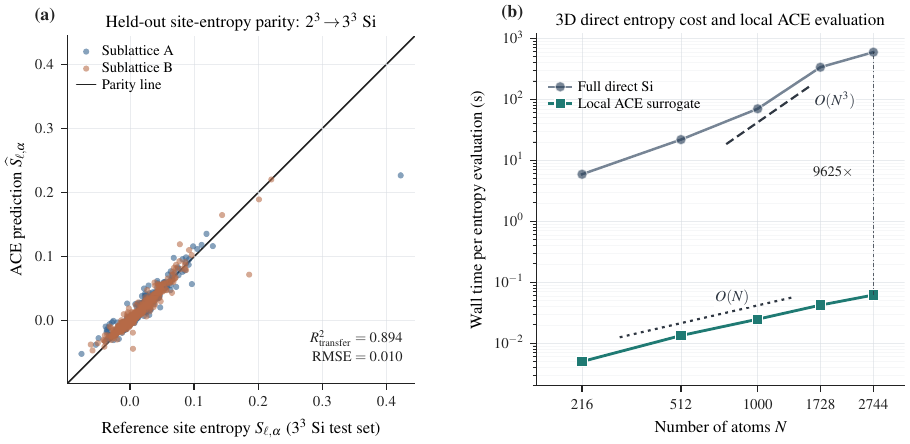}
\caption{Si SW: cross-supercell transfer and 3D timing. (a) Held-out site-entropy parity for the ACE site model on an independent $3\times 3\times 3$ Si test set ($940$ held-out site labels from $20$ configurations, $R^2_{\rm transfer}=0.894$ with $95\%$ block-bootstrap confidence interval $[0.867,0.943]$), trained on the $2\times 2\times 2$ data set of \cref{sec:numerics:si-sw}. (b) Measured full direct Si timings and fixed-cutoff local ACE evaluation for $N\in\{216,512,1000,1728,2744\}$ atoms. The dashed and dotted segments show cubic and linear reference envelopes, respectively.}
\label{fig:si-scaling-3d}
\end{figure}

\section{Conclusions}
\label{sec:conclusions}

This paper addresses the cost of repeated harmonic vibrational-entropy evaluations in multilattices. The central result is that, for finite-range or screened atomistic models satisfying full block stability and the admissible perturbation condition, the global spectral entropy can be decomposed into sublattice-resolved site labels with controlled spatial locality. Compact-core finite-cell-limit site entropies have entropy-gradient decay $(1+|\ell-n|)^{-2d}$ and the associated cutoff-local site model has truncation error of order $\rcut^{-d}$. For non-compact dyadically admissible fields, the same finite-cell kernel estimate applies when the additional periodic-kernel convergence is available. The analytical step that makes the multilattice case possible is the block resolvent estimate, which transfers the acoustic, mixed, and optical Green-function hierarchy through the homogeneous normalization without losing the Bravais entropy-locality rate.

The numerical experiments support the resulting surrogate calculation. Controlled 1D and 2D multilattice tests show the predicted locality and cutoff behaviour with stability-dependent prefactors, and the 3D test follows the predicted envelope over the dense-diagonalisation sizes that can be computed directly. On Stillinger--Weber Si, a sublattice-resolved ACE site model learns harmonic site-entropy labels on both randomized-site and configuration-level splits and transfers to an independent larger supercell. On zinc-blende CdTe, the species-label ablation shows that collapsing inequivalent chemical sublattices removes information needed for accurate local entropy prediction. The timing tests demonstrate the intended computational transition: after one-time Hessian label generation on training configurations, fixed-cutoff surrogate evaluation scales linearly in the number of atoms and is orders of magnitude faster than repeated dense diagonalisation in the tested range.

These results make harmonic entropy a practical local quantity for repeated finite-temperature defect calculations, rather than a global spectral calculation that must be redone from scratch for every configuration. The most immediate applications are formation entropies of vacancies and anti-sites in ordered alloys, attempt-frequency ratios along migration paths in Vineyard transition-state theory~\cite{Vineyard1957FrequencyFactors}, and harmonic finite-temperature corrections to dislocation core or interface energies. In these settings, a trained site model can be reused across supercell sizes and sampled configurations, while the cutoff and regression errors remain explicit parts of the calculation.

The scope is intentionally bounded. The present theorem is harmonic and finite-range, or screened-effective, in character; fully anharmonic free energies require quasi-harmonic, SCHA-type, or thermodynamic-integration extensions~\cite{Hellman2013SCAILD,Errea2014SCHA,KirkwoodTI1935,Frenkel1984FreeEnergy,Grabowski2009AnharmonicFreeEnergy}. The constants deteriorate near soft optical modes, loss of full block stability, or rough coefficient fields outside the admissible class. Unscreened Coulomb interactions, dipolar tails, and non-analytic polar phonon effects are also outside the finite-range theorem and call for a separate long-range splitting or screened-operator analysis. Extending the same locality-and-surrogate approach to those settings is the natural next step toward entropy-aware defect calculations in broader multispecies materials.

\section*{Declaration of Competing Interest}
The authors declare that they have no known competing financial interests or personal relationships that could have appeared to influence the work reported in this paper.

\section*{Data and Code Availability}
The implementation uses ASE~\cite{LarsenASE2017}, \texttt{LAMMPS}~\cite{ThompsonLAMMPS2022}, and \texttt{ACEpotentials.jl}~\cite{Witt2023ACEpotentials,Lysogorskiy2021PerformantImplementationACE}. The code, data sets, potential/parameter files, and run instructions needed to reproduce the numerical experiments will be released in a versioned GitHub/Zenodo archive before acceptance. Wall-time measurements (\cref{sec:numerics:scaling}) were performed on an Apple M3 CPU using NumPy~1.26 with Accelerate BLAS, ASE~3.23, and Julia~1.11 with \texttt{ACEpotentials.jl}~v0.9.

\appendix
\section{Supporting Estimates for the Locality Result}
\label{app:proofs}

This appendix provides the technical estimates used by the locality and truncation results in the main text. The proofs are organized around five steps: regularizing the homogeneous acoustic singularity, recording the borderline lattice convolution estimate, closing the admissible dyadic chain, separating compact-core and far-field stiffness perturbations, and applying the resulting chain bound to the normalized resolvent. The main text states these estimates with proof ideas so that the computational argument can be read independently; the details below make the decay constants, zero-mode removal, and admissibility requirements explicit.

\Cref{tab:bravais-vs-multi} gives a short roadmap of where the multilattice proof differs from the Bravais argument of~\cite{TorabiWangOrtner2024Entropy}. The final algebraic rates of $\partial S/\partial u$ and of the truncation error are unchanged; the additional work is in the block structure and in the constants.

\begin{table}[h]
\centering
\caption{Proof roadmap comparing the Bravais analysis of~\cite{TorabiWangOrtner2024Entropy} with the multilattice estimates used here. The additional estimates enter through the acoustic--optical block structure and the normalized block resolvent.}
\label{tab:bravais-vs-multi}
\small
\begin{tabular}{@{}>{\raggedright\arraybackslash}p{0.27\textwidth}>{\raggedright\arraybackslash}p{0.30\textwidth}>{\raggedright\arraybackslash}p{0.36\textwidth}@{}}
\toprule
Object & Bravais (\cite{TorabiWangOrtner2024Entropy}) & Multilattice (this work) \\
\midrule
Homogeneous symbol $\widehat\Hhom(\xi)$ & scalar; $\widehat\Hhom(\xi)\sim A_{\rm a}|\xi|^2$, $\xi\to 0$ &
$|\Scal|d\times|\Scal|d$ matrix; block form
$\bigl(\begin{smallmatrix}\mathbb{C}_{\rm a}|\xi|^2 & \mathbb{C}_{\rm ap}|\xi| \\ \mathbb{C}_{\rm ap}^*|\xi| & \mathbb{C}_{\rm p}\end{smallmatrix}\bigr)$ \\[0.2em]
Acoustic--optical coupling & none & $O(|\xi|)$ off-diagonal block $\mathbb{C}_{\rm ap}$ \\[0.2em]
Inverse-square-root symbol $\widehat\Fhom(\xi)$ near $\xi=0$ & scalar $|\xi|^{-1}$ &
block:
$O(|\xi|^{-1})$ acoustic, $O(1)$ mixed, $O(1)$ optical; the three orders must be regularized separately \\[0.2em]
Atomistic Green's function decay & single rate $|G(\ell)|\lesssim(1+|\ell|)^{-(d-2)}$ &
hierarchy~\cite[Thm.~3.6]{OlsonOrtnerWangZhang2023MultilatticeDislocations}: $G_{00}\!\sim\!(1+|\ell|)^{-(d-2)}$, $G_{0p}\!\sim\!(1+|\ell|)^{-(d-1)}$, $G_{pp}\!\sim\!(1+|\ell|)^{-d}$ \\[0.2em]
Regularization by $\mathcal{D}$ & single $|\xi|$ factor cancels the scalar $|\xi|^{-1}$ pole &
intra-sublattice differences are $O(|\xi|)$ on all blocks; sublattice-shift legs are $O(|\xi|)$ on $P_0$ and bounded on $P_p$; \emph{which} bound applies depends on the block \\[0.2em]
Kernel of $\mathcal{D}\Fhom^*$ & $(1+|\ell-m|)^{-d}$ scalar &
$(1+|\ell-m|)^{-d}$ uniformly across all four blocks (\cref{lem:hom-kernel}) \\[0.2em]
Resolvent of $T(u)$ along contour & scalar Neumann series with $r^{-d}$ kernel &
matrix-valued Neumann series; block resolvent estimate needed (\cref{lem:block-resolvent}) \\[0.2em]
Locality rate of $\partial S/\partial u$ & $(1+r)^{-2d}$ & $(1+r)^{-2d}$ (\textbf{unchanged})\\[0.2em]
Truncation rate of $S_{\ell,\alpha}$ & $\rcut^{-d}$ & $\rcut^{-d}$ (\textbf{unchanged})\\[0.2em]
Stability constants entering $C_{\rm loc}$ & $\gamma_a$ & block margin, optical gap, coupling/projector conditioning, $|\Scal|$, and admissibility constants \\
\bottomrule
\end{tabular}
\end{table}

\subsection{Proof of \texorpdfstring{\cref{lem:hom-kernel}}{Lemma (Homogeneous normalized kernel)}}
\label{app:hom-kernel}

The lattice Fourier transform on the index $\ell\in\Lambda$ gives the matrix-valued symbol $\widehat\Hhom(\xi)$ of \eqref{eq:fourier-symbol}--\eqref{eq:hom-symbol-block} on the Brillouin zone $\BZ$. The block square-root analysis in \cref{lem:block-sqrt-symbol} gives the acoustic/optical expansion of $\widehat\Fhom(\xi)=\widehat\Hhom(\xi)^{-1/2}$ and the Marcinkiewicz estimates for all finite-difference-regularized products used below.

Each finite-difference leg $\mathcal{D}$ entering the Hessian stencil decomposes (cf.\ \cref{sec:setting}) into two types of building blocks:
\begin{enumerate}
  \item an intra-sublattice difference $D_j\psi(\ell)=\psi(\ell+j)-\psi(\ell)$, with Fourier symbol $\widehat{D}_j(\xi)=e^{-i\xi\cdot j}-1=O(|\xi|)$ as $\xi\to 0$;
  \item a sublattice-shift difference $\widetilde D_{j;\alpha\beta}=T_jE_{\alpha\beta}-E_{\alpha\alpha}$, where $E_{\alpha\beta}$ is the cell-internal matrix unit. Its acoustic restriction is $(e^{-i\xi\cdot j}-1)E_{\alpha\beta}P_0=O(|\xi|)$, and its internal-shift restriction to $P_p$ is uniformly bounded.
\end{enumerate}
Multiplying $\widehat{\mathcal{D}}(\xi)$ on the left of $\widehat\Fhom(\xi)$ regularizes the acoustic singularity in \eqref{eq:Fhat-block}: every Hessian leg contributes the needed $|\xi|$ factor on $P_0$, while the $P_p$ columns of $\widehat\Fhom$ are already bounded. The product symbol $\sigma(\xi):=\widehat{\mathcal{D}}(\xi)\,\widehat\Fhom(\xi)^*$ is therefore bounded on $\BZ$ and $C^\infty$ on $\BZ\setminus\{0\}$. Near $\xi=0$, $\sigma$ is homogeneous of degree zero up to smooth corrections: its acoustic part has leading form $(i\xi\cdot j)|\xi|^{-1}A_{\rm a}(\hat\xi)^{-1/2}$ with $\hat\xi=\xi/|\xi|$, and its optical part is bounded. As a consequence,
\begin{equation}
  |\partial^\rho_\xi\sigma(\xi)|\leq C_\rho\,|\xi|^{-|\rho|}
  \qquad\text{on }\BZ\setminus\{0\},
  \label{eq:sigma-deg-zero}
\end{equation}
for every multi-index $\rho\in\mathbb{N}_0^d$, with constants depending on $\gamma_a$, $\gamma_{\rm p}$, $\rcut^{(\rm hom)}$, $|\Scal|$, and $d$. At $\xi=0$ the symbol $\sigma$ is bounded but typically discontinuous: the ratio $\widehat{\mathcal{D}}(\xi)/|\xi|$ has a unit-vector limit along radial directions, so $\sigma$ has a Riesz-type conical singularity rather than a smooth extension. The standard integration-by-parts estimate therefore loses control at the origin and must be combined with a near-zero scaling argument.

Fix $r:=|\ell-m|\geq 1$ and split the Brillouin-zone integral at the scale $\eta=1/r$. With a smooth cutoff $\chi_\eta\in C_c^\infty(\BZ)$ satisfying $\chi_\eta=1$ on $B_\eta(0)$, $\chi_\eta=0$ off $B_{2\eta}(0)$, and $|\partial^\rho_\xi\chi_\eta|\leq C\eta^{-|\rho|}$,
\begin{equation}
  [\mathcal{D}\Fhom^*]_{(\ell,\alpha),(m,\gamma)}
  =\frac{1}{(2\pi)^d}\!\int_\BZ\! e^{i\xi\cdot(\ell-m)}\sigma(\xi)\dd\xi
  =I_{\rm near}(\ell-m)+I_{\rm far}(\ell-m),
  \label{eq:near-far-split}
\end{equation}
with $I_{\rm near}$ carrying the cutoff $\chi_\eta\sigma$ and $I_{\rm far}$ carrying $(1-\chi_\eta)\sigma$.

For $I_{\rm near}$, the symbol $\chi_\eta\sigma$ is supported in $B_{2\eta}(0)$ where $|\sigma|\leq C$ (the regularization absorbed the singularity), so the trivial bound gives
\begin{equation}
  |I_{\rm near}(\ell-m)|\leq C\,\mathrm{vol}(B_{2\eta})\leq C'\eta^d=C'r^{-d}.
  \label{eq:near}
\end{equation}

For $I_{\rm far}$, $(1-\chi_\eta)\sigma$ is supported away from $\xi=0$, so $d+1$ successive integrations by parts in the direction $\hat e:=(\ell-m)/r$ are admissible and produce
\begin{equation}
  |I_{\rm far}(\ell-m)|
  \leq
  C\,r^{-(d+1)}\!\int_{|\xi|\geq\eta}\!\big|\partial^{d+1}_\xi\big[(1-\chi_\eta(\xi))\sigma(\xi)\big]\big|\dd\xi.
  \label{eq:far-ibp}
\end{equation}
We bound the integrand by a dyadic decomposition. Write the support $\{\eta\leq|\xi|\leq c_\BZ\}$ as a disjoint union of shells $A_j:=\{2^{-j-1}c_\BZ<|\xi|\leq 2^{-j}c_\BZ\}$ for $j=0,\dots,J$ with $2^{-J}c_\BZ\sim\eta$. On $A_j$ the regularized symbol satisfies $|\partial^{d+1}_\xi\sigma(\xi)|\leq C\,2^{j(d+1)}c_\BZ^{-(d+1)}$ by \eqref{eq:sigma-deg-zero}, and derivatives hitting $\chi_\eta$ are supported only on $A_J$ with bound $C\eta^{-(d+1)}\mathbf{1}_{A_J}$. Since $\mathrm{vol}(A_j)\sim 2^{-jd}c_\BZ^d$,
\begin{equation}
  \int_{A_j}|\partial^{d+1}_\xi[(1-\chi_\eta)\sigma]|\dd\xi
  \leq C\,2^{j(d+1)}c_\BZ^{-(d+1)}\cdot 2^{-jd}c_\BZ^d
  = C\,c_\BZ^{-1}\,2^{j},
\end{equation}
and the geometric series in $j$ from $0$ to $J$ telescopes to $C\,c_\BZ^{-1}\,2^{J}\leq C'\,\eta^{-1}=C'\,r$. The boundary term on $A_J$ contributes $C\eta^{-(d+1)}\mathrm{vol}(A_J)\leq C'\eta^{-1}=C'\,r$ likewise. Substituting into \eqref{eq:far-ibp},
\begin{equation}
  |I_{\rm far}(\ell-m)|\leq C\,r^{-(d+1)}\cdot r = C\,r^{-d}.
  \label{eq:far}
\end{equation}
This is the standard symbol-of-degree-zero estimate, in which performing $d+1$ rather than $d$ integrations by parts closes the dyadic geometric series into a single power of $r$ without a logarithmic factor. Combining \eqref{eq:near} and \eqref{eq:far} gives \eqref{eq:hom-kernel} with $C=C(\gamma_a,\gamma_{\rm p},\rcut^{(\rm hom)},|\Scal|,d)$ independent of $\ell,m$ and the supercell size; the estimate is uniform for $r\geq 1$, and the case $r=0$ is the trivial bound $|\sigma|_{L^1(\BZ)}<\infty$. The estimate for $\Fhom\mathcal{D}^*$ is the adjoint statement and follows from $\widehat\Fhom(\xi)^*=\widehat\Fhom(\xi)$. The finite-cell version is the periodic Riemann sum approximation of the same Fourier integral, with the singular point $\xi=0$ excluded by the finite-cell zero-mode projection.\hfill$\square$

\subsection{Proof of \texorpdfstring{\cref{lem:block-conv}}{Lemma (Block lattice convolution)}}
\label{app:block-conv}

The block matrix product is
\begin{equation}
  [K_1K_2]_{(\ell,\alpha),(m,\gamma)}
  =\sum_{p\in\Lambda}\sum_{\delta\in\Scal}[K_1]_{(\ell,\alpha),(p,\delta)}\,[K_2]_{(p,\delta),(m,\gamma)},
\end{equation}
each summand bounded by $C_1C_2(1+|\ell-p|)^{-a_1}(1+|p-m|)^{-a_2}$ by \eqref{eq:block-conv-input}. The sublattice sum produces a factor $|\Scal|$, absorbed into the constant. The remaining lattice sum is split into the regions $|p-\ell|\leq r/2$, $|p-m|\leq r/2$, and the complement, where $r=|\ell-m|$. The two endpoint regions give the integrable-side contributions $(1+r)^{-a_2}\sum_{|q|\leq r/2}(1+|q|)^{-a_1}$ and $(1+r)^{-a_1}\sum_{|q|\leq r/2}(1+|q|)^{-a_2}$; the middle region gives $(1+r)^{d-a_1-a_2}$. Combining these three bounds yields \eqref{eq:block-conv-output}, with the logarithm in the critical case $a_1=a_2=d$. The estimate transfers verbatim to finite periodic cells provided the lattice distance is replaced by its periodic counterpart and the supercell separation $L_N$ exceeds $|\ell-m|$.\hfill$\square$

\subsection{Proof of \texorpdfstring{\cref{lem:chain-algebra}}{Lemma (Cancellation-sensitive chain algebra)}}
\label{app:chain-algebra}

The homogeneous kernels in \cref{lem:chain-algebra} are translation-invariant before the coefficient factors $A_{q_j}$ are inserted. Their Fourier symbols satisfy the Marcinkiewicz bounds
\begin{equation}
  |\partial_\xi^\rho \sigma(\xi)|\leq C_\rho |\xi|^{-|\rho|},
  \qquad \xi\neq0,
  \label{eq:chain-marcinkiewicz}
\end{equation}
after the finite-difference cancellations have been included; this is exactly the estimate proved in \eqref{eq:sigma-deg-zero} for $\mathcal{D}\Fhom^*$ and in \eqref{eq:block-marcinkiewicz} for $\mathcal{D}G^{\rm hom}\mathcal{D}^*$. In the revised resolvent order all endpoint factors in \cref{lem:chain-algebra} are degree-zero regularized kernels: $L_{q_1}=\mathcal{D}G^{\rm hom}\mathcal{D}_{q_1}^*$, the interior factors are $\mathcal{D}_{q_j}G^{\rm hom}\mathcal{D}_{q_{j+1}}^*$, and $R_{q_k}=\mathcal{D}_{q_k}\Fhom$.

We first record the compact-core branch because it requires no frequency cancellation. If each $A_q$ is supported in a fixed core $B_{R_{\rm core}}$, all intermediate summation indices in the chain lie in a fixed finite set. The left endpoint gives $(1+\dist(\ell,{\rm core}))^{-d}$, the right endpoint gives $(1+\dist(m,{\rm core}))^{-d}$, and every interior core-to-core factor is a finite matrix with operator norm bounded by the homogeneous constants and $A$. Hence the chain is bounded by
\[
  C(C_{\rm core}A)^k
  (1+\dist(\ell,{\rm core}))^{-d}
  (1+\dist(m,{\rm core}))^{-d}
  \leq C(C_{\rm core}A)^k(1+|\ell-m|)^{-d},
\]
with constants independent of the periodic cell. This proves \eqref{eq:chain-algebra} for compact-core perturbations, after absorbing $C_{\rm core}$ into $C_{\rm hom}$.

It remains to treat the non-compact dyadically admissible branch. Choose a smooth dyadic partition of unity $\{\psi_j\}_{j\geq0}$ on $\BZ\setminus\{0\}$, where shell $j$ has frequency scale $2^{-j}$ and spatial scale $2^j$, and write each homogeneous factor $K$ as $K=\sum_{j\ge0}K_j$ with Fourier symbol $\psi_j\sigma$. For the finite periodic grid the same decomposition is used on the nonzero Brillouin points; the zero point is absent because all operators are restricted to the stable subspace. We use three properties of these shells. First, for a degree-zero regularized factor $K$, the shell kernel $K_j$ satisfies
\begin{equation}
  |K_j(n)|\leq C\,2^{-jd}\big(1+2^{-j}|n|\big)^{-(d+1)},
  \qquad
  \sum_{n\in\Lambda}K_j(n)=0
  \label{eq:dyadic-cz}
\end{equation}
on the zero-frequency-regularized sector; in the finite-cell setting the zero-frequency contribution is absent after the translational zero-mode removal. The zero moment follows because $\psi_j$ is supported away from $\xi=0$ and hence the Fourier value at $\xi=0$ is zero. Second, the finite number of block and bond labels is absorbed into a geometric factor depending on $|\Scal|$ and the stencil but not on $k$. Third, the coefficient factors are used through the admissibility estimate \eqref{eq:dyadic-admissibility}; this is the additional hypothesis that replaces the invalid assertion that arbitrary Schur-bounded multipliers preserve dyadic scale.

The cancellation in \eqref{eq:dyadic-cz}, together with \eqref{eq:dyadic-admissibility}, gives the almost-orthogonality estimate underlying this chain bound,
\begin{equation}
  \|K_i A K_j'\|_{\ell^1\to\ell^\infty}
  \leq C A\,2^{-c|i-j|}\,2^{-d\max\{i,j\}},
  \label{eq:almost-orthog}
\end{equation}
for any two adjacent regularized homogeneous pieces in the chain, with constants uniform over the finitely many block labels and over the admissible coefficient class. This is exactly the place where \cref{def:admissible-perturbation} is used; no comparable estimate is asserted for arbitrary lattice-scale oscillatory multipliers.

Expanding a chain of length $k$ gives a sum over dyadic indices $i_0,\ldots,i_k$ associated with the left endpoint, the $k-1$ interior kernels, and the right endpoint. The product of adjacent estimates contains
\[
  \prod_{j=0}^{k-1}2^{-c|i_j-i_{j+1}|}.
\]
The matrix $S_{ij}:=2^{-c|i-j|}$ has uniformly bounded row and column sums on $\ell^1(\mathbb{N}_0)$ and on $\ell^\infty(\mathbb{N}_0)$. Therefore repeated summation over $i_1,\ldots,i_{k-1}$ contributes at most $\|S\|^{k}$, a geometric factor that is absorbed into $(C_{\rm hom}A)^k$, rather than a factor depending on the number of dyadic shells. This proves the almost-orthogonal scale summation uniformly in the chain length and in the finite periodic grids, where the number of shells grows with the supercell.

After the scale indices have been summed, one top-scale kernel remains. Summing its off-diagonal envelope from \eqref{eq:dyadic-cz} gives
\[
  \sum_{j\geq0}2^{-jd}\big(1+2^{-j}|\ell-m|\big)^{-(d+1)}
  \lesssim (1+|\ell-m|)^{-d},
\]
which proves \eqref{eq:chain-algebra}. The finite periodic version is obtained by the same dyadic decomposition on the discrete Brillouin grid, with the zero mode removed by the stable-subspace projection and constants independent of the number of nonzero grid shells.\hfill$\square$

\subsection{Proof of \texorpdfstring{\cref{lem:localized-perturbation}}{Lemma (Localized normalized perturbation)}}
\label{app:localized-perturbation}

By the assumed finite-range structure of the multilattice site potential, the harmonic Hessian decomposes as
\begin{equation}
  \Hess(u)=\Hhom+\delta\Hess(u),
  \qquad
  \delta\Hess(u)=\sum_q \mathcal{D}_q^*\,A_q(u)\,\mathcal{D}_q,
\end{equation}
where the index $q$ runs over the finite set of bond legs and $A_q(u)$ is the bond-stiffness perturbation associated with leg $q$, i.e. the corresponding stiffness block at displacement $u$ minus its homogeneous value. By construction, $\Fhom^*\Hhom\Fhom=\Id$ on the zero-frequency-regularized sector, hence
\begin{equation}
  T(u)-\Id
  =\Fhom^*\delta\Hess(u)\Fhom
  =\sum_q (\mathcal{D}_q\Fhom)^*\,A_q(u)\,(\mathcal{D}_q\Fhom),
\end{equation}
which is \eqref{eq:localized-perturbation}.

\paragraph{Compact-core coefficient perturbation.}
If $A_q(u)$ is supported within a finite radius $R_{\rm core}$ of the defect core, \cref{lem:hom-kernel} bounds the kernel of each leg $\mathcal{D}_q\Fhom$ uniformly,
\begin{equation}
  \big|[\mathcal{D}_q\Fhom]_{(p,\sigma),(\ell,\alpha)}\big|\leq C(1+|p-\ell|)^{-d}.
  \label{eq:DqF-bound}
\end{equation}
Composing two such kernels with the localized middle factor $A_q(u)$ gives, for any $(\ell,\alpha),(m,\gamma)$,
\begin{align}
  \big|[T(u)-\Id]_{(\ell,\alpha),(m,\gamma)}\big|
  &\leq \sum_q\sum_{p_1,p_2} \big|[\mathcal{D}_q\Fhom]_{(p_1,\cdot),(\ell,\alpha)}\big|\,\|A_q(u)\|\,\big|[\mathcal{D}_q\Fhom]_{(p_2,\cdot),(m,\gamma)}\big|
  \nonumber\\
  &\leq C^2\|A(u)\|_\infty
       \sum_{p_1,p_2\in B_{R_{\rm core}}({\rm core})}(1+|p_1-\ell|)^{-d}(1+|p_2-m|)^{-d}
  \nonumber\\
  &\lesssim (1+\dist(\ell,{\rm core}))^{-d}(1+\dist(m,{\rm core}))^{-d},
\end{align}
where the last step uses that for $\ell,m$ outside $B_{2R_{\rm core}}({\rm core})$ each factor $(1+|p_i-\ell|)^{-d}\lesssim(1+\dist(\ell,{\rm core}))^{-d}$ and the inner sum over $p_1,p_2$ is bounded by $|B_{R_{\rm core}}|^2$, a constant; for $\ell$ inside $B_{2R_{\rm core}}({\rm core})$ the bound $(1+\dist(\ell,{\rm core}))^{-d}\geq c>0$ is trivial. This proves \eqref{eq:localized-perturbation-decay} in the compact-core regime.

\paragraph{Algebraic far-field coefficient perturbation.}
For a relaxed defect with an elastic far field, $u$ is defined on all of $\Lambda$ and the bond-stiffness perturbation is no longer compactly supported: along each bond leg $q$ at base point $p_q$,
\begin{equation}
  \|A_q(u)\|\lesssim\|\partial^3 V\|\,|Du(p_q)|,
\end{equation}
and applying \eqref{eq:DqF-bound} on both legs of the representation \eqref{eq:localized-perturbation} gives
\begin{equation}
  \big|[T(u)-\Id]_{(\ell,\alpha),(m,\gamma)}\big|
  \leq C\,\|\partial^3 V\|\sum_{q\in\Lambda}|Du(p_q)|\,(1+|p_q-\ell|)^{-d}(1+|p_q-m|)^{-d}.
  \label{eq:elastic-step2}
\end{equation}
Suppose the elastic far field decays as $|Du(p)|\lesssim(1+\dist(p,{\rm core}))^{-(d-1+\beta)}$ for some $\beta>0$. Then \eqref{eq:elastic-step2} is precisely the convolution bound \eqref{eq:farfield-conv} with $a(p)\lesssim |Du(p)|$. This is the bound used for finite-cell far-field estimates and for the conditional non-compact regime described in the main text. It should not be simplified to the product \eqref{eq:localized-perturbation-decay}: if $\ell$ and $m$ are both far from the core and $p$ is near $\ell\approx m$, the two kernel factors are $O(1)$ and the contribution scales like $|Du(\ell)|$, which is generally only $(1+\dist(\ell,{\rm core}))^{-(d-1+\beta)}$. Thus the compact-core product estimate and the far-field convolution estimate are separate statements.\hfill$\square$

\subsection{Proof of \texorpdfstring{\cref{lem:block-resolvent}}{Lemma (Block normalized resolvent)}}
\label{app:block-resolvent}

At the homogeneous configuration $T(0)=\Id$ on the zero-frequency-regularized sector, so the unperturbed resolvent is $R^{\rm hom}_z=(z-1)^{-1}\Id$. We choose a contour $C$ with
\begin{equation}
  \dist(C,1)= r_0,\qquad r_0>\|T(u)-\Id\|_{\ell^2\to\ell^2},
  \label{eq:contour-condition}
\end{equation}
which is admissible whenever $u\in\mathcal{B}_{\delta_*}$ by \cref{def:stab-nbhd}. On $C$, the Neumann series
\begin{equation}
  R_z=\sum_{k=0}^{\infty}(z-1)^{-(k+1)}\big(T(u)-\Id\big)^k
  \label{eq:neumann}
\end{equation}
converges in operator norm with summable bounds.

We estimate $(\mathcal{D}\Fhom)R_z$; the second bound in \eqref{eq:block-resolvent} follows by adjoint symmetry. The term $k=0$ is $\mathcal{D}\Fhom$, which is controlled by \cref{lem:hom-kernel}. For $k\geq1$, substituting \eqref{eq:localized-perturbation} gives
\begin{align}
  (\mathcal{D}\Fhom)\big(T(u)-\Id\big)^k
  ={}&\sum_{q_1,\dots,q_k}
  \big(\mathcal{D}G^{\rm hom}\mathcal{D}_{q_1}^*\big)A_{q_1}
  \Big[\prod_{j=1}^{k-1}
  \big(\mathcal{D}_{q_j}G^{\rm hom}\mathcal{D}_{q_{j+1}}^*\big)A_{q_{j+1}}\Big]
  \big(\mathcal{D}_{q_k}\Fhom\big).
  \label{eq:chain}
\end{align}
Indeed, the adjacent homogeneous factors combine through
\begin{equation}
  (\mathcal{D}_{q_j}\Fhom)(\mathcal{D}_{q_{j+1}}\Fhom)^*
  =\mathcal{D}_{q_j}G^{\rm hom}\mathcal{D}_{q_{j+1}}^*.
  \label{eq:DGD-block}
\end{equation}
The Marcinkiewicz estimate \eqref{eq:block-marcinkiewicz} and the near/far argument of \cref{app:hom-kernel} give, uniformly over the finite set of bond and block labels,
\begin{equation}
  \big|[\mathcal{D}_{q_j}G^{\rm hom}\mathcal{D}_{q_{j+1}}^*]_{(\ell,\alpha),(m,\gamma)}\big|
  \leq C_{\rm DGD}(1+|\ell-m|)^{-d},
  \label{eq:DGD-rate}
\end{equation}
and the same rate holds for the endpoint $\mathcal{D}G^{\rm hom}\mathcal{D}_{q_1}^*$, while $\mathcal{D}_{q_k}\Fhom$ is controlled by \cref{lem:hom-kernel}. Thus \eqref{eq:chain} is exactly the chain covered by \cref{lem:chain-algebra}. If $A:=\max_q\|A_q(u)\|_{\ell^2\to\ell^2}$, then
\begin{equation}
  \big|[(\mathcal{D}\Fhom)(T(u)-\Id)^k]_{(\ell,\alpha),(m,\gamma)}\big|
  \leq C_{\rm ch}(C_{\rm hom}A)^k(1+|\ell-m|)^{-d}.
  \label{eq:Neumann-term}
\end{equation}
The admissibility of the coefficients is used only through \cref{lem:chain-algebra}; without it, the critical absolute convolution would produce a logarithmic loss.

Inserting \eqref{eq:Neumann-term} into \eqref{eq:neumann} and using $|z-1|=r_0$ on $C$ gives the scalar majorant
\[
  C_{\rm ch}\sum_{k\geq0}r_0^{-(k+1)}(C_{\rm hom}A)^k.
\]
This is the extra smallness requirement that is not supplied by operator-norm convergence alone. By the definition of the stability neighbourhood, see \eqref{eq:kernel-smallness}, $C_{\rm hom}A\leq\theta r_0$ with $\theta<1$, and hence the majorant is bounded by $C_{\rm ch}r_0^{-1}(1-\theta)^{-1}$. Therefore
\begin{equation}
  \big|[(\mathcal{D}\Fhom)R_z]_{(\ell,\alpha),(m,\gamma)}\big|
  \leq C_{\rm res}\,(1+|\ell-m|)^{-d},
\end{equation}
uniformly in $z\in C$. The adjoint relation
\[
  \big(R_z(\Fhom^*\mathcal{D}^*)\big)^*
  =(\mathcal{D}\Fhom)R_{\bar z}
\]
gives the second estimate in \eqref{eq:block-resolvent}. The finite-cell version is obtained by replacing the Brillouin-zone integrals in \cref{lem:hom-kernel} by their periodic Riemann-sum analogues, with the zero mode removed by the finite-cell projection. This is the resolvent input used in the contour formula for \cref{thm:locality}.\hfill$\square$

\bibliographystyle{cas-model2-names}
\bibliography{refs}

\end{document}